%% file: main.tex
\documentclass[a4paper,anonymous,11pt]{article}
\usepackage[margin=1in]{geometry}
\usepackage{titling}
\usepackage{times}

\usepackage[utf8]{inputenc}

\long\def\commentbegin #1\commentend{}

\usepackage{balance}
\usepackage{algorithm}
\usepackage[noend]{algpseudocode}
\usepackage{ifthen}
\usepackage{comment}
\usepackage{amsthm,amsmath}
\usepackage{amssymb}
\usepackage{graphicx}
\usepackage{hyperref}

\usepackage{pstricks,pst-node,pst-tree}
\usepackage{soul}

\usepackage{color}
\usepackage{mathrsfs}
\usepackage[normalem]{ulem}
\usepackage{paralist}

\renewcommand{\epsilon}{\varepsilon}

\newcommand{\Prob}[1]{\hbox{\rm I\kern-2pt P}\left[#1\right]}
\def\polylog{\operatorname{polylog}}
\def\poly{\operatorname{poly}}

\DeclareMathAlphabet{\mathsc}{OT1}{cmr}{m}{sc}

\newtheorem{theorem}{Theorem}
\newtheorem{lemma}{Lemma}

\newtheorem{observation}{Observation}

\renewcommand{\geq}{\geqslant}
\renewcommand{\ge}{\geqslant}
\renewcommand{\leq}{\leqslant}
\renewcommand{\le}{\leqslant}

\newcommand{\I}{\mathcal{I}}

\newcommand{\disj}{{\mathsf{Disj}}}

\renewcommand{\P}{{\mathcal{P}}}

\newboolean{short}
\setboolean{short}{false}

\algblockdefx[ExecBl]{BlockOn}{BlockOff}  [1]{#1}
  
\makeatletter
\ifthenelse{\equal{\ALG@noend}{t}}%
  {\algtext*{BlockOff}}
  {}%
\makeatother

\newcommand{\shortOnly}[1]{\ifthenelse{\boolean{short}}{#1}{}}
\newcommand{\onlyShort}[1]{\ifthenelse{\boolean{short}}{}{#1}}
\newcommand{\longOnly}[1]{\ifthenelse{\boolean{short}}{}{#1}}
\newcommand{\onlyLong}[1]{\ifthenelse{\boolean{short}}{}{#1}}

\def\ShowComment{True}
\ifdefined\ShowComment
\def\billy#1{{\color{green}\underline{\textsf{Billy:}}} {\color{blue} \emph{#1}}}
\def\gopal#1{{\color{red}\underline{\textsf{Gopal:}}} {\color{blue} \emph{#1}}}
\def\john#1{{\color{orange}\underline{\textsf{John:}}} {\color{blue} \emph{#1}}}
\else
\def\billy#1{}
\def\gopal#1{}
\def\john#1{}
\fi

\def\inline#1:{\par\vskip 7pt\noindent{\bf #1:}\hskip 10pt}

\def\TS{Transmission-Schedule}

\def\RANDMST{\textsc{Randomized-MST}}
\def\UPCASTMIN{\textsc{Upcast-Min}}
\def\DOWNCAST{\textsc{Fragment-Broadcast}}
\def\TRANSMITADJACENT{\textsc{Transmit-Adjacent}}

\def\TRANSMITNEIGHBOR{\textsc{Transmit-Neighbor}}
\def\MSTTRADEOFF{\textsc{Trade-Off-MST}}

\def\DETMST{\textsc{Deterministic-MST}}
\def\AWAKECOLORING{\textsc{Fast-Awake-Coloring}}
\def\NEIGHBORAWARE{\textsc{Neighbor-Awareness}}
\def\REORIENTFRAG{\textsc{Merging-Fragments}}

\def\BUILDMSC{\textsc{Build-MSC}}
\def\SAF{\textsc{SAF}}
\def\SIMPLELE{\textsc{SIMPLE-LE}}
\def\SIMPLEBFS{\textsc{SIMPLE-BFS}}
\def\LEBFS{\textsc{LE-BFS}}

\def\DOWNRECEIVE{\textsf{Down-Receive}}
\def\DOWNSEND{\textsf{Down-Send}}
\def\SIDESENDRECEIVE{\textsf{Side-Send-Receive}}
\def\UPRECEIVE{\textsf{Up-Receive}}
\def\UPSEND{\textsf{Up-Send}}

\def\IMPPHASE{\mathcal{P}}
\def\DETIMPPHASE{\mathcal{P}}
\def\NUMFRAGSDET{c}

\def\LOCAL{\ensuremath{\mathcal{LOCAL}}}

\def\sd{\textsf{SD}}
\def\dsd{\textsf{DSD}}
\def\css{\textsf{CSS}}
\def\mst{\textsf{MST}}
\def\INCOMINGMOE{\textsf{INCOMING-MOE}}
\def\NEWLEVELNUM{\textsf{NEW-LEVEL-NUM}}
\def\NEWFRAGID{\textsf{NEW-FRAGMENT-ID}}
\def\NBRINFO{\textsf{NBR-INFO}}

\renewcommand{\geq}{\geqslant}
\renewcommand{\ge}{\geqslant}
\renewcommand{\leq}{\leqslant}
\renewcommand{\le}{\leqslant}

\long\def\hide #1\hideend{}

\newcommand{\squishlist}{
 \begin{list}{$\bullet$}
  { \setlength{\itemsep}{0pt}
     \setlength{\parsep}{3pt}
     \setlength{\topsep}{3pt}
     \setlength{\partopsep}{0pt}
     \setlength{\leftmargin}{1.5em}
     \setlength{\labelwidth}{1em}
     \setlength{\labelsep}{0.5em} } }

\newcommand{\squishlisttwo}{
 \begin{list}{$\bullet$}
  { \setlength{\itemsep}{0pt}
     \setlength{\parsep}{0pt}
    \setlength{\topsep}{0pt}
    \setlength{\partopsep}{0pt}
    \setlength{\leftmargin}{2em}
    \setlength{\labelwidth}{1.5em}
    \setlength{\labelsep}{0.5em} } }

\newcommand{\squishend}{
  \end{list}  }

\def\CONGEST{CONGEST}
\def\LOCAL{LOCAL}

\title{Awake Complexity of Distributed Minimum Spanning Tree\thanks{Part of the work was done while the William K. Moses Jr. was a Post Doctoral Fellow at the University of Houston.\\
J. Augustine was supported, in part, by DST/SERB MATRICS Grant MTR/2018/001198 and the Centre of Excellence in Cryptography Cybersecurity and Distributed Trust under the IIT Madras Institute of Eminence Scheme and by the VAJRA visiting faculty program of the Government of India.\\
W. K. Moses Jr. was supported, in part, by NSF grants CCF-1540512, IIS-1633720, and CCF-1717075 and BSF grant 2016419.\\
G. Pandurangan was supported, in part, by NSF grants CCF-1540512, IIS-1633720, and CCF-1717075 and BSF grant 2016419 and by the VAJRA visiting faculty program of the Government of India.
}
}

\author{John Augustine$^1$ \and William K. Moses Jr.$^2$ \and Gopal Pandurangan$^3$}
\date{%
    $^1$Indian Institute of Technology Madras, Chennai, India\\%
    $^2$Durham University, Durham, UK\\%
    $^3$University of Houston, Houston, TX, USA\\%
}

\begin{document}
\maketitle
\thispagestyle{empty}

\begin{abstract}
    \input{abstract}

\end{abstract}

\thispagestyle{empty}

\newpage

\setcounter{page}{1}

\section{Introduction}\label{sec:intro}
\input{intro}

\section{Lower Bounds}\label{sec:lower-bound}
\input{lower-bound}

\section{MST Algorithms with Optimal Awake Complexity}\label{sec:algorithms}
\input{algorithms}

\section{MST Algorithm with a Trade-off}\label{sec:mst-trade-offs}
\input{mst-trade-offs}

\section{Conclusion}\label{sec:conclusion}
\input{conclusion}

~\\
\noindent \textbf{Acknowledgments:} We thank the anonymous reviewer for the useful idea that helped to reduce the run time of the deterministic awake-optimal algorithm, in particular, the coloring procedure. We thank Fabien Dufoulon for helpful discussions.

\bibliographystyle{plainurl}
\bibliography{references,reference,biblio}

\end{document}

%% file: abstract.tex
The \emph{awake complexity} of a distributed algorithm  measures  the number of rounds in which a node is awake. When a node is not awake, it is {\em sleeping} and does not do any computation or communication and spends very little resources.
Reducing the awake complexity of a distributed algorithm
can be relevant in resource-constrained networks such as sensor networks, where saving energy of nodes is crucial. 
Awake complexity of many fundamental problems such as maximal independent set, maximal matching, coloring,
and spanning trees have been studied recently.
    
In this work, we study  the awake complexity of the fundamental distributed  minimum spanning tree (MST) problem and present the following results.
\begin{itemize}
\item {\bf Lower Bounds.} 
\begin{enumerate}
\item We show a lower bound of $\Omega(\log n)$ (where $n$ is the number of nodes in the network) on the awake complexity 
for computing an MST  that holds even for randomized algorithms.
\item To better understand the relationship between the awake complexity  and the  round complexity (which counts
both awake and sleeping rounds),
we also prove a \emph{trade-off}  lower bound of 
$\tilde{\Omega}(n)$\footnote{Throughout, the $\tilde{O}$
notation hides a $\text{polylog } n$  factor and $\tilde{\Omega}$ hides a $1/(\text{polylog } n)$ factor.} 
on the
product of round complexity and awake complexity
for any distributed algorithm (even randomized) that outputs an MST. Our lower bound is shown for  graphs having diameter $\tilde{\Omega}(\sqrt{n})$. 
\end{enumerate}
\item {\bf Awake-Optimal Algorithms.} 
\begin{enumerate} \item We present a distributed randomized algorithm to find an MST that achieves the
 optimal awake complexity of $O(\log n)$ (with high probability).  
Its round complexity is $O(n \log n)$ and by our trade-off lower bound,  this 
is the best  round complexity (up to logarithmic factors) for an awake-optimal algorithm.
\item We also show that the $O(\log n)$ awake complexity bound can be achieved deterministically as well, by presenting
a distributed \emph{deterministic} algorithm that has $O(\log n)$ awake complexity and $O(n \log^5 n)$ round complexity. We also show how to reduce the round complexity to $O(n \log n \log^* n)$ at the expense of a slightly increased awake complexity of $O(\log n \log^* n)$. 
\end{enumerate}
\item {\bf Trade-Off Algorithms.} 
To complement our trade-off lower bound, we present a parameterized family of distributed algorithms 
that gives an essentially optimal trade-off (up to $\text{polylog } n$ factors) between the awake complexity and the round complexity. Specifically we
show a family of distributed algorithms that find an MST of the given graph with high probability in $\tilde{O}(D + 2^k + n/2^k)$ round complexity and $\tilde{O}(n/2^k)$  awake complexity, where $D$ is the network diameter and 
integer $k$ is an input parameter to the algorithm. When $k \in [\max \lbrace \lceil 0.5\log n \rceil, \lceil \log D \rceil \rbrace, \lceil \log n \rceil]$, we can obtain useful trade-offs.
\end{itemize}

Our work is a step towards understanding resource-efficient distributed algorithms for fundamental global 
problems such as MST. It shows that MST can be computed with any node being awake (and hence spending resources) for only $O(\log n)$ rounds  which is significantly better than the  fundamental lower bound of $\tilde{\Omega}(\text{Diameter}(G)+\sqrt{n})$ rounds for MST in the traditional CONGEST model, where nodes can be active for at least so many rounds.

%% file: intro.tex

We study the distributed minimum spanning tree (MST) problem, a central problem in
distributed computing. This problem has been studied extensively for several decades
starting with the seminal work of Gallagher, Humblet, 
and Spira (GHS) in the early 1980s \cite{DistMst:Gallager}; for example, we refer to
the survey of \cite{eatcs} that traces the history of the problem till the state of the art.
The round (time) complexity of the GHS algorithm is $O(n \log n)$ rounds,
where $n$ is the number of nodes in the network.
The round complexity of the problem has been continuously improved since then  and now tight optimal bounds are known.\footnote{Message complexity has also been well-studied, see e.g., \cite{eatcs}, but this is not the focus of this paper, although  our algorithms are also (essentially) message optimal --- see Section \ref{sec:results}.} 
It is now well-established
that $\Theta(D+\sqrt{n})$ is essentially (up to logarithmic factors)  a tight  bound for
the round complexity of distributed MST in the standard CONGEST model~\cite{peleg-bound,elkin-bound,kutten-domset,stoc11}. 
The lower bound  applies even to randomized Monte Carlo algorithms~\cite{stoc11}, while  deterministic algorithms that match this bound (up to logarithmic factor) are now well-known (see e.g., 
\cite{peleg,kutten-domset,PanduranganRS17,elkin17}).
Thus, the round complexity of the problem in the traditional CONGEST distributed model is settled (see also the recent works
of \cite{universal-optimality-mst,low-congestion-mst}). In the CONGEST model, any node can send, receive, or do local computation in any round and only $O(\log n)$-sized
messages can be sent through any edge per round.

MST serves
as a basic primitive in  many network applications including efficient broadcast, leader election, approximate Steiner tree construction etc.~\cite{peleg, eatcs}. For example, an important application of MST is for
energy-efficient broadcast in wireless networks which has been extensively studied, see e.g., \cite{ambuhl,khan-tpds}.\footnote{Note that for energy-efficient broadcast, the underlying graph is {\em weighted} and it
is known that using an MST  to broadcast minimizes the total cost\cite{ambuhl,khan-tpds}.}

In resource-constrained networks such as  sensor networks, where nodes  spend a lot of energy or other resources over the course
of an algorithm, a round complexity 
of $\tilde{O}(D+\sqrt{n})$ (which is essentially optimal) to construct an MST can be large. In particular, in such an algorithm, a node can be active
over the entire course of the algorithm.
It is worth studying whether MST can be constructed in such a way that each node is active only in a small number of rounds ---
much smaller than that taken over the (worst-case) number of rounds --- and hence could  use much less resources such as energy.

Motivated by the above considerations, in recent years several works have studied energy-distributed algorithms (see e.g., \cite{energy1,CDHHLP18, CDHP20, podc2020,BM21,King_2011}). 
This paper studies the distributed MST problem in the 
\emph{sleeping model}~\cite{podc2020}.
In this model (see Section~\ref{sec:model}),  nodes can
operate in two modes: \emph{awake} and \emph{sleeping}. Each node can choose to
enter the awake or asleep state at the start of any specified round.
In the sleeping mode, a node cannot send, receive, or do any local computation;
messages sent to it are also lost. The resources utilized in sleeping rounds
are negligible and hence only awake rounds are counted. The goal in the sleeping
model is to design distributed algorithms that solve problems in a small number
of awake rounds, i.e., have small \emph{awake complexity} (also called {\em awake time}), which is the
(worst-case) number of awake rounds needed by any node until it terminates. This
is motivated by the fact that, if the awake complexity is small, then every node
takes only a small number of rounds during which it uses a significant amount of
resources. For example, in ad hoc wireless or sensor networks,
a node's energy consumption depends on the amount of time it is actively communicating with nodes.
In fact, significant amount of energy is spent by a node even when it is just  waiting to hear from a neighbor~\cite{podc2020}.
On the other hand, when a node is sleeping --- when all of its radio devices are switched off --- it spends little or no energy. 
While the main goal is to minimize awake complexity, we would also like to
minimize the (traditional) \emph{round complexity} (also called {\em time complexity or run time}) of the algorithm, which counts
the (worst-case) total number of rounds taken by any node, including both awake
and sleeping rounds.

The  work of Barenboim and Maimon~\cite{BM21} shows that global problems such as broadcast and constructing a (arbitrary) spanning
tree (but not an MST on weighted graphs, which is required for applications mentioned earlier) can be accomplished in $O(\log n)$ awake rounds in the sleeping (CONGEST) model (see also the related result of \cite{CDHHLP18} ---  Section \ref{sec:related}.).  The above work is significant because it shows that even such global problems 
can be accomplished in a very small number of awake rounds, bypassing the $\Omega(D)$ lower bound on the round complexity (in the traditional model). In this work, we focus on another fundamental global problem, namely MST, and study its awake complexity. We show  that MST  can be solved in  $O(\log n)$ rounds awake
complexity which we show is also {\em optimal}, by showing a matching {\em lower bound}. The upper bound is in contrast to the 
classical $\tilde{\Omega}(D+\sqrt{n})$ lower bound on the round complexity (in the traditional CONGEST model) where nodes can be awake for at least so many rounds.
Another key issue we study is the relationship between  awake complexity and round complexity of MST. It is intriguing  whether
one can obtain an algorithm that has both optimal awake and round complexity. We answer this in the negative by showing a {\em trade-off} lower bound (see Section \ref{sec:results}).

\subsection{Distributed Computing Model and Complexity Measures}
\label{sec:model}

\paragraph*{Distributed Network Model}
We are given a distributed network 
modeled as an arbitrary, undirected, connected, weighted graph $G(V,E,w)$, 
where the node set $V$ ($|V| = n$) represent the processors, the edge set $E$ ($|E| = m$) represents the communication
links between them, and $w(e)$ is the weight of edge $e \in E$. 

The network diameter is denoted by $D$, also called the hop-diameter (that is, the unweighted
diameter) of $G$, and in this paper by diameter we always mean hop-diameter.
We  assume that the weights of the edges of
the graph are all distinct. This implies that the MST of the graph is unique.
(The definitions and the results generalize readily to the case where the weights are not necessarily distinct.)

Each node hosts a processor with limited initial knowledge.
We assume that nodes have unique \texttt{ID}s, 
and at the beginning of the computation each node is provided its \texttt{ID} as input and the weights
of the edges incident to it. We assume 
node IDs are of size $O(\log n)$ bits.   
We assume that each node has ports (each port having a unique port number); 
each incident edge is connected to one distinct port. 
We also assume that nodes know $n$, the number of nodes in the network. For the deterministic algorithm, we make the additional assumption that nodes know the value of $N$, an upper bound on the largest ID of any node. 
Nodes initially do not have any other global knowledge and have knowledge of only themselves.

Nodes are allowed to communicate through the edges of the graph $G$ and it is assumed that communication is {\em synchronous} and occurs in rounds. In particular, we assume that each node knows the current round number, starting from round 1. 
In each round, each node can perform some local computation (which happens instantaneously) 
including accessing a private source of randomness, and can exchange (possibly distinct) $O(\log{n})$-bit messages 
with each of its neighboring nodes. This is the traditional {\em synchronous CONGEST} model.

As is standard, the goal, at the end of the distributed MST computation, is for every node to know which of its incident edges belong to the MST.

\paragraph*{Sleeping Model and Complexity Measures}
 The \emph{sleeping model}~\cite{podc2020} is a generalization of the traditional model, where  a node can be in
 either of the two states --- sleeping or awake --- before it finishes executing the algorithm (locally). (In the traditional model, each node is always awake until it finished the algorithm). Initially, we assume that all nodes are awake.
 That is, any node $v$, can decide to {\em sleep} starting at any (specified) round of its choice; we assume all nodes know the correct round number whenever they are awake. It can {\em wake up} again later at any specified round and enter the {\em awake} state.
 In the sleeping state, a node does not send or receive messages, nor does it do any local computation. Messages sent to it by other nodes when it was sleeping are lost. This aspect makes it especially challenging to design algorithms
 that have a small number of awake rounds, since one has to carefully coordinate the transmission of messages.

  Let $A_v$ denote the number of awake rounds  for a node
    $v$ before termination. We define the \emph{(worst-case) awake complexity}
     as $\max_{v \in V}A_v$. For a randomized algorithm, $A_v$ will be a random variable
     and our goal is to obtain high probability bounds on the awake complexity.
     Apart from minimizing the awake complexity, we also strive to minimize
     the overall (traditional) \emph{round complexity} 
     (also called \emph{run time or time complexity}), where both, sleeping and awake rounds, are counted.

\subsection{Our Contributions and Techniques}
\label{sec:results}
We study the awake complexity of distributed MST and present the following  results (see Table~\ref{table:results} for a summary). 

\input{tables}

\noindent {\bf 1. Lower Bound on the Awake Complexity of MST.}
We  show that $\Omega(\log n)$ is a lower bound 
on the awake complexity of constructing a MST, even for randomized Monte Carlo algorithms 
 with constant success probability (Section \ref{sec:awake-lower}). We note that showing lower bounds on the awake
 complexity is different from showing lower bounds for round complexity in the traditional LOCAL model. In the traditional LOCAL
 model, {\em locality} plays an important role in showing lower bounds. In particular, to obtain
 information from a $r$-hop neighborhood, one needs at least $r$ rounds and lower bounds use indistinguishability
 arguments of identical $r$-hop neighborhoods to show a lower bound of $r$ to solve a problem. For example, 
 this approach is used to show a lower bound of $\Omega(D)$ for leader election or broadcast even for randomized algorithms
 \cite{jacm15}. Lower bounds in the CONGEST model are more involved and exploit bandwidth restriction
 to show stronger lower bounds for certain problems, e.g., MST has a  round complexity lower bound of $\tilde{\Omega}(D+\sqrt{n})$ rounds. In contrast, as we show in this paper, MST can be solved using only $O(\log n)$
 awake rounds and requires a different type of lower bound argument for awake complexity.

 We first show a  randomized lower bound of $\Omega(\log n)$ awake complexity for  {\em broadcast} on a line of $n$ nodes. 
 Our broadcast lower bound is an adaptation of a similar lower bound shown for the  energy complexity model \cite{CDHHLP18}.
 Our randomized lower bound is more general and subsumes  a deterministic lower bound
 of $\Omega(\log n)$ shown in \cite{BM21} for computing a spanning tree called the Distributed Layered Tree (by a reduction this lower bound applies to broadcast and leader election as well). The deterministic lower bound follows  from the {\em deterministic  message complexity} lower bound of $\Omega(n \log n)$  on
 leader election on rings (which also requires an additional condition that node IDs should be from a large range) to show
 that some node should be awake at least $\Omega(\log n)$ rounds. Note that this argument does not immediately apply for randomized algorithms, since such a message complexity lower bound does not hold for randomized (Monte Carlo) leader election~\cite{jacm15}. On the other hand, our lower bound uses  probabilistic arguments to show
 that some node has to be awake for at least $\Omega(\log n)$ rounds for accomplishing broadcast in a line. We then use a reduction to argue the same lower bound for MST problem on a {\em weighted ring}. We believe that the broadcast lower bound  is fundamental to awake complexity and will have implications for several other problems as well.

\noindent {\bf 2. Lower Bounds Relating Awake and Round Complexities.}
An important question is whether one can  optimize awake complexity and round complexity {\em simultaneously} or whether one can only optimize one at the cost of the other. 
Our next result shows that the latter is true by showing a (existential) lower bound on the {\em product} of awake and round complexities. Specifically, we construct a family of graphs with diameters ranging between $\tilde{\Omega}(\sqrt{n})$ to $\tilde{O}(n)$. Time-optimal MST algorithms for these graphs will have round complexity within $\polylog{n}$ factors of their diameters. We show that any distributed algorithm  on this graph family that requires only $\tilde{O}(\mbox{Diameter}(G))$ rounds must have an awake complexity of at least $\tilde{\Omega}(n/\mbox{Diameter}(G))$ for any distributed  algorithm. This holds even for Monte-Carlo randomized algorithms
with constant success probability. The precise result is stated in
Theorem \ref{thm:lowerbound} (see Section \ref{sec:tradeoff-lower}).  In other words, the product
of the round complexity and awake complexity is $\tilde{\Omega}(n)$. We note that this product lower  bound is shown for graphs with
diameter  $\tilde{\Omega}(\sqrt{n})$, and we leave open whether a similar bound holds for graphs with much smaller diameters.

Our lower bound technique for showing a conditional lower bound on the awake complexity (conditional
on upper bounding the round complexity)  can be of independent
interest. We use a lower bound graph family that is similar to that used in prior work (e.g., \cite{peleg-bound, stoc11}),
but our lower bound technique uses  communication complexity to lower bound awake complexity by lower bounding
the \emph{congestion caused in some node}. This is different from the \emph{Simulation Theorem} technique (\cite{stoc11}) used
to show unconditional lower bounds for {\em round complexity} in the traditional setting. The main idea  is showing that  to solve  the distributed set disjointness problem 
in less than $c$ rounds, at least $\tilde{\Omega}(n/c)$ bits have to send through some node that has small (constant) degree. This means that the node has to be awake
for essentially $\tilde{\Omega}(n/c)$  rounds. By using standard reductions, the same lower bound holds for MST.
The  technique is quite general and could  be adapted to show similar conditional lower bounds on the awake complexity
for other fundamental problems such as shortest paths, minimum cut etc.

\noindent {\bf 3. Awake-Optimal Algorithms and Techniques.}
We present a distributed randomized algorithm (Section \ref{sec:mst-optimal-awake-time}) that has $O(\log n)$ awake complexity which is optimal since it matches the lower bound shown.     
The round complexity of our algorithm is $O(n \log n)$ and by our trade-off lower bound,  this 
is the best  round complexity (up to logarithmic factors) for an awake-optimal algorithm.
We then show that the awake-optimal bound of $O(\log n)$ can be obtained deterministically by presenting a deterministic
algorithm (see Section \ref{sec:det-mst-opt-awake-time}). However, the deterministic algorithm  has a slightly higher round complexity
of $O(n \log^5 n)$, assuming that the node IDs are in the range $[1,N]$ and $N =O(\poly{n})$ is known to all nodes.
We also show that one can  reduce the deterministic round complexity to $O(n \log n \log^* n)$ at the cost of
slightly increasing the awake complexity to $O(\log n \log^* n)$ rounds.

Our algorithms use several techniques  for constructing an MST in an awake-efficient manner. Some of these can be of
independent interest in designing such algorithms for other fundamental  problems.\footnote{As an example, the $O(n \log n \log^* n)$ deterministic algorithm is   useful in designing an MIS
algorithm with small awake and round complexities \cite{DMP23} --- see Section \ref{sec:related}.} 
Our main idea is to construct a spanning tree called the \emph{Labeled Distance Tree (LDT)}. 
An LDT is a rooted oriented spanning tree such that each node is labeled by its  distance from the root and every node knows the labels of its parent  and children (if any) and the label of the root of the tree. We show
that an LDT can be constructed in $O(\log n)$ awake complexity and in 
a weighted graph, the LDT can be constructed so that  it is an MST.  While LDT construction is akin to the classic GHS/Boruvka algorithm~\cite{DistMst:Gallager}, it is technically challenging to construct an LDT that is also an MST in $O(\log n)$ awake rounds (Section \ref{sec:algorithms}). As in GHS algorithm, starting with $n$ singleton fragments we merge fragments
via minimum outgoing edge (MOE) in each phase. The LDT distance property allows for finding an MOE in $O(1)$ awake rounds, since broadcast and convergecast can be accomplished in 
$O(1)$ rounds. On the other hand,  maintaining the property of LDT when fragments are merged is non-trivial. 
A key idea is that we show that merging can be accomplished
in $O(1)$ awake rounds by merging only fragment chains of {\em constant} length (details under ``Technical challenges'' in Sections \ref{sec:mst-optimal-awake-time} and \ref{sec:det-mst-opt-awake-time}).\footnote{Consider the supergraph where the fragments are nodes and the MOEs are edges. A fragment chain is one such supergraph that forms a path. The exact details of the supergraphs formed are slightly different and explained in the relevant section, but this idea is beneficial to understanding.}
We develop a technical lemma that shows that this merging restriction still reduces the number of fragments
by a constant factor in every phase and hence the algorithm takes overall $O(\log n)$ awake rounds.

Our tree construction is  different compared to the construction of trees in \cite{BM21,CDHHLP18}.  In particular, a tree structure called as Distributed Layered Tree (DLT) is used in Barenboim and Maimon~\cite{BM21}.
 A DLT is a rooted oriented spanning tree where the vertices are labeled, such that each
vertex has a greater label than that of its parent, according to a given order and each vertex knows its own label and the label of its parent. Another similar tree structure is used in~\cite{CDHHLP18}. 
A crucial difference between the LDT construction and the others is that it allows fragments to be merged via
desired edges (MOEs), unlike the construction of DLT, for example, where one has to merge along edges that connect a higher label to a lower label. This is not useful for MST construction. Another important difference is that the labels used in LDTs scale with the number of nodes, whereas the labels in DLTs scale with the maximum ID assigned to any node. As the running time for both constructions are proportional to the maximum possible labels and the ID range is usually polynomially larger than the number of nodes, the running time to construct a DLT is much larger than the running time to construct an LDT. 

\noindent {\bf 4. Trade-Off Algorithms.} 
We present a parameterized family of distributed algorithms 
that show a trade-off between the awake complexity and the round complexity and essentially (up to a $\polylog n$ factor)  matches
our product lower bound of $\tilde{\Omega}(n)$.\footnote{The product lower bound of
$\tilde{\Omega}(n)$ is shown  for graphs with diameter at least $\tilde{\Omega}(\sqrt{n})$. Hence, the near tightness claim holds for graphs in this diameter range.} Specifically we
show a family of distributed algorithms that find an MST of the given graph with high probability in $\Tilde{O}(D + 2^k + n/2^k)$ running time and $\Tilde{O}(n/2^k)$  awake time, where $D$ is the network diameter and 
integer $k \in [\max \lbrace \lceil 0.5\log n \rceil, \lceil \log D \rceil \rbrace, \lceil \log n \rceil]$ is an input parameter to the algorithm. Notice that when $D=O(\sqrt{n})$, the round complexity can vary from $\Tilde{O}(\sqrt{n})$ to $\Tilde{O}(n)$, and  we can choose (integer) $k \in [ \lceil 0.5\log n \rceil, \lceil \log n \rceil]$ from  to get the  $\tilde{O}(n)$ product bound for this entire range. On the other hand, when $D = \omega(\sqrt{n})$, the round complexity can vary  from $\tilde{O}(D)$ to $\Tilde{O}(n)$, and we can choose $k \in  [ \lceil \log D \rceil , \lceil \log n \rceil]$ and get a family of algorithms  with (essentially) optimal round complexities from $\Tilde{O}(D)$ to $\Tilde{O}(n)$.

\subsection{Related Work and Comparison}
\label{sec:related}

The sleeping model and the awake complexity measure was introduced in a paper by
Chatterjee, Gmyr and Pandurangan~\cite{podc2020} who showed that MIS in general graphs can be
solved in $O(1)$ rounds \emph{node-averaged} awake complexity. Node-averaged
awake complexity is measured by the \emph{average} number of rounds a node is
awake.  
The (worst-case) awake complexity of their MIS algorithm is $O(\log n)$, while the
worst-case complexity (that includes all rounds, sleeping and awake) is
$O(\log^{3.41}n)$ rounds.
Subsequently, Ghaffari and Portmann~\cite{ghaffari-sleeping}  developed a randomized MIS algorithm
that has worst-case complexity of $O(\log n)$, while having $O(1)$ node-averaged awake complexity (both
bounds hold with high probability). They studied  approximate maximum matching
and vertex cover and presented algorithms that have similar node-averaged and worst-case awake complexities.
These results show that the above fundamental local symmetry breaking problems  have $O(\log n)$
(worst-case) awake complexity as is shown for global problems such as spanning tree \cite{BM21} and MST (this paper).
In a recent result, Dufoulon, Moses Jr., and Pandurangan~\cite{DMP23}
show that MIS can be solved in $O(\log \log n)$ (worst-case)
awake complexity which is exponentially better than previous results. But the round complexity is $O(poly(n))$.
It then uses  the \emph{deterministic} LDT construction algorithm
of this paper to obtain an MIS algorithm that has a slightly larger awake complexity of $O(\log \log n \log^* n)$,
but significantly better round complexity of $O(\polylog{n})$. The existence of a deterministic LDT algorithm
is crucial to obtaining their result.

Barenboim and Maimon~\cite{BM21} showed that many problems, including broadcast, construction of a spanning tree,
 and leader election can be solved deterministically in $O(\log n)$ awake complexity. They also showed
 that fundamental symmetry breaking problems such as MIS and ($\Delta+1$)-coloring can be solved deterministically
 in $O(\log \Delta + \log^*n)$ awake rounds in the \LOCAL\ model, where $\Delta$ is the maximum degree. 
More generally, they also define  the class of {\em O-LOCAL} problems (that includes MIS and coloring) and 
showed that problems in this class admit  a deterministic algorithm that runs in $O(\log \Delta + \log^*n)$ awake time and $O(\Delta^2)$ round complexity.
 Maimon \cite{maimon} presents trade-offs between awake and round complexity for O-LOCAL problems.

While there is significant amount of work on energy-efficient distributed algorithms over the years  we discuss
those that are most relevant to this paper.
 A recent line of relevant  work is  that Chang, Kopelowitz, Pettie, Wang, and Zhan and their follow ups \cite{energy1,CDHHLP18,CDHP20,energy4,DH22} (see also the references therein and its follow up papers mentioned below and also  much earlier work on energy complexity in radio networks e.g., \cite{Nakano,jurdin,pajak}).
 This work defines
the measure of {\em energy complexity} which is the same as (worst-case) awake complexity (i.e., both measures count only the rounds that a node is awake).  While the awake complexity used here and several other papers \cite{podc2020,ghaffari-sleeping,BM21} assumes the usual
\CONGEST (or \LOCAL) communication model (and hence the model can be called \emph{SLEEPING-CONGEST} (or \emph{SLEEPING-LOCAL})), the energy complexity  measure used in \cite{energy1} (and also  papers  mentioned above) has some additional communication restrictions
that pertain to radio networks (and can be called \emph{SLEEPING-RADIO} model). 
 The most important being that nodes can only
broadcast messages (hence the same message is sent to all neighbors) and when
a node transmits, no other neighboring node can. (Also a node cannot transmit and listen in the same round.) The energy model has a few
variants depending on how collisions are handled. There is a version of the SLEEPING-RADIO model called 
``Local''  where collisions are ignored and nodes can transmit messages at
the same time; this is essentially same as SLEEPING-LOCAL model, apart from the
notion that in a given round a node can transmit only the same message to its neighbors. In particular, upper  bounds  in the radio model 
apply directly to the sleeping model. 
In particular, we use a recent result due to Dani and Hayes~\cite{DH22} that computes breadth-first search (BFS) tree
with $O(\polylog(n))$ energy complexity in the radio model as a subroutine in our MST tradeoff algorithm. 
Also, algorithms in the SLEEPING-CONGEST model can be made to work in the SLEEPING-RADIO model
yielding similar bounds (with possibly  a $O(\polylog(n))$ multiplicative factor) to the energy/awake complexity. 

Lower bounds shown in the local version of the SLEEPING-RADIO model apply
to other models including SLEEPING-LOCAL (and SLEEPING-CONGEST). For example, Chang, Dani, Hayes, He, Li, and Pettie~\cite{CDHHLP18} show a lower bound 
$\Omega(\log n)$ on the energy complexity
of broadcast which applies also to randomized algorithms. 
This lower bound is shown for the local version of their model,
and  this result holds also for the awake complexity in the sleeping model. 
We adapt this lower bound result to show a $\Omega(\log n)$ lower bound on the awake complexity of MST even for randomized algorithms.

%% file: tables.tex
\begin{table*}[ht]
\footnotesize
	\caption{
	Summary of our Results.
	} 
	\centering 
		\resizebox{1.0\columnwidth}{!}{%
	\begin{tabular}{|c|c|c|c|c|c|}
		\hline
		Algorithm & Type & Awake Time (AT) & Run Time (RT) & AT Lower Bound  & AT $\times$ RT Lower Bound \\
		\hline
		\hline
		*$\RANDMST$ & Randomized & $O(\log n)$ & $O(n \log n)$ & $\Omega(\log n)$ & $\tilde{\Omega}(n)\dagger$\\
		\hline
		$\DETMST$ & Deterministic & $O(\log n)$ & $O(n\log^5 n)$ & $\Omega(\log n)$ & $\tilde{\Omega}(n)\dagger$\\
		\hline
		Modified $\DETMST$ & Deterministic &  $O(\log n \log^* n)$ & $O(n \log n \log^* n)$ & $\Omega(\log n)$ & $\tilde{\Omega}(n)\dagger$\\
		\hline
		*$\clubsuit\MSTTRADEOFF$ & Randomized & $\Tilde{O}(n/2^k)$ & $\Tilde{O}(D + 2^k + n/2^k)$ & - & -\\
		\hline
		\multicolumn{6}{|l|}{*The algorithm outputs an MST with high probability.}\\
		\multicolumn{6}{|l|}{$\clubsuit$This algorithm takes integer $k$ as an input parameter.}\\
        \multicolumn{6}{|l|}{$\dagger$Holds for (some) graphs with diameter $\tilde{\Omega}(\sqrt{n})$.}\\
		\multicolumn{6}{|l|}{Our lower bounds also apply to Monte Carlo randomized algorithms with constant success probability.}\\
		\multicolumn{6}{|l|}{$n$ is the number of nodes in the network and $N$ is an upper bound on the largest ID of a node.}\\
		\hline
	\end{tabular}
		}
	\label{table:results}
\end{table*}

%% file: lower-bound.tex

\subsection{Unconditional Lower bound on Awake Complexity of MST}
\label{sec:awake-lower}

We consider a ring of $\Theta(n)$ nodes with random weights and edges. The two largest weighted edges will be apart by a hop distance of $\Omega(n)$ with constant probability and any MST algorithm must detect which one has the lower weight. Clearly, this will require communication over either one of the $\Omega(n)$ length paths between the two edges. Under this setting, we get the following theorem.

\begin{theorem}
\label{thm:lb-unconditional}
Any algorithm to solve MST with probability exceeding $1/8$ on a ring network comprising $\Theta(n)$ nodes requires $\Omega(\log n)$ awake time even when message sizes are unbounded.
\end{theorem}

\begin{proof}
Consider a weighted ring $R$ of length $4n+4$ with each node $u$ having a random ID denoted $ID(u)$ and every edge $(u,v)$ having a random weight $w(u,v)$, with all IDs and weights being drawn uniformly and independently from a sufficiently large $\poly(n)$ space. The IDs and weights will be distinct whp, so we condition on this for the rest of the proof. 

Let $\P$ be any protocol executed by the nodes in the ring. For the purpose of this lower bound, we assume that $\P$ operates as follows. Initially, every node $u$ generates a sufficiently large random string $r_u$ and uses $r_u$ to make all its random choices throughout the execution of $\P$. 
Every pair of neighboring vertices $u$ and $v$ that are simultaneously awake in some round $r$ communicate their states to each other in their entireties. Thus, each node $u$ starts with an initial knowledge comprising just its own ID\footnote{Normally, we assume $KT_1$ knowledge (i.e.,
each node knows the IDs of its neighbors) at the beginning, but here, we assume $KT_0$ (i.e., no such knowledge) without loss of generality, because we can always add a single awake round for all nodes at the beginning to go from $KT_0$ to $KT_1$.}, weights of incident edges and its own random string. More generally, at the end of its $a$th awake round, $u$ knows everything about a segment $S(u,a)$ (including weights of edges incident on end vertices of $S(u,a)$) but nothing outside it. 

Let $\I_k$ denote the family of connected segments of length $k$.   For each $I \in I_k$, we define an event $U(I,a)$ as follows: there exists a special vertex $v^* \in I$ such that $S(v^*, a) \subseteq I$. We now show a crucial lemma that will help us complete the proof.

\begin{lemma}\label{lem:induct}
Let $a_{max} = \lfloor \log_{13} n \rfloor$. For every $a \in \{0, 1, \ldots,a_{max}\} = [a_{max}]$ and $k=13^a$ and $I \in I_k$, $\Pr(U(I,a)) \ge 1/2$. Moreover, for any two non-overlapping segments $I_1$ and $I_2$ (i.e., with no common vertex), the events $U(I_1, a)$ and $U(I_2, a)$ are independent.
\end{lemma}

\begin{proof}
We prove this by induction on $a$. Clearly, $U(I, 0)$ holds for every $I \in \I_1$, thus establishing the basis. Now consider any $a> 0$ and any $I \in \I_k$ for $k = 13^a$. We split $I$ into 13 non-overlapping segments 
of length $13^{a-1}$ each. The probability that $U(*, a-1)$ will occur for five of the 13 segments is at least 5/6, thanks to the independence of $U(*, a-1)$. Now let us focus on the five of the 13 segments (denoted $A$, $B$, $C$, $D$, and $E$ taken, say, in clockwise order) for which the event $U(*, a-1)$ occurred with $v(A)$, $v(B)$, \ldots, $v(E)$ denoting the special vertex. Again, due to independence, the time when each  $v(*)$ wakes up for the $a$th time are independent. Thus, with probability at least 3/5, $v(B)$, $v(C)$, or $v(D)$ will wake up for their respective $a$th time no later than $v(A)$ or $v(E)$. Therefore, with probability at least $(3/5)(5/6) = 1/2$, one of $v(B)$, $v(C)$, or $v(D)$ would  become the special vertex pertaining to $U(I,a)$, thereby ensuring that $\Pr(U(I,a) \ge 1/2)$. 
\end{proof}

Let $e_1=(u_1,v_1)$ and $e_2=(u_2,v_2)$ be the two edges with the largest weights. With probability at least 1/2, they are separated by a distance of $n+1$.  The MST for this ring comprises all edges other than either $e_1$ or $e_2$, whichever is weighted more. Thus, any MST algorithm for this ring must be able to decide between the two edges. This requires communication either between $u_2$ and $v_1$ or between $v_2$ and $u_1$. This will require (with probability at least $1/4$) $\Omega(\log n)$ awake time for at least one node along either paths. 
\end{proof}

We begin by showing that $\Omega(\log n)$ is  an {\em unconditional} lower bound on the awake time for MST. This shows that our algorithms presented in Section \ref{sec:algorithms} achieve optimal awake complexity.  We then
show a  lower bound  of $\tilde{\Omega}(n)$ on the product of the awake and round complexities. This can be considered
as a conditional lower bound on awake complexity, conditioned on an upper bound on the round complexity. This conditional lower bound shows
that our randomized awake optimal algorithm (see Section \ref{sec:mst-optimal-awake-time}) has essentially  the best possible round complexity (up to a $\polylog{n}$ factor). 

\subsection{Lower Bound on the Product of Awake and Round Complexity}
\label{sec:tradeoff-lower}
We adapt the lower bound technique from~\cite{DHKKNPPW11} to the sleeping model and show a lower bound on the product of the awake and round complexities, thereby exposing an inherent trade off between them. 

The high level intuition is as follows. Note that both endpoints of an edge must be awake in a round for $O(\log n)$ bits to be transmitted across in that round. When only one of the endpoints is awake, the awake node can sense the other node is asleep, which is $O(1)$ bits of information. There is no transfer of information when both endpoints are asleep. Thus, if an edge $e = (u,v)$ must transmit $B$ bits when executing an algorithm in the CONGEST model, then, either  $u$ or $v$ must be awake for at least $\Omega(B/\log n)$ rounds. Thus congestion increases awake time. 
Our goal is to exploit this intuition to prove the lower bound. 

We use a communication complexity based reduction to reduce {\em set disjointness} (\sd) in the classical communication complexity model to the {\em distributed set disjointness} (\dsd) problem in the sleeping model and then extend \dsd\ to the minimum spanning tree problem via an intermediate {\em connected spanning subgraph} (\css), both to be solved in the sleeping model.  

In the \sd\ problem, two players, Alice and Bob, possess two $k$-bit strings $a=(a_i)_{1 \le i \le k}$ and $b=(b_i)_{1 \le i \le k}$, respectively. They are required to compute an output bit $\disj(a,b)$ that is 1 iff there is no $i \in [k]$ such that $a_i = b_i = 1$ (i.e., the inner product $\langle a, b \rangle = 0$), and 0 otherwise. Alice and Bob must compute $\disj(a,b)$  while exchanging the least number of bits. It is well-known that any protocol that solves \sd\ requires $\Omega(k)$ bits (on expectation) to be exchanged between Alice and Bob even if they employ a randomized protocol~\cite{Razborov92} that can fail with a small fixed probability $\epsilon>0$. 

Note that the $\Omega(k)$ lower bound (on the number of bits exchanged) for \sd\ problem applies to the asynchronous setting which is the default model in classical 2-party communication complexity. On the other hand, we assume the synchronous setting for the sleeping model. However, it can be shown using the Synchronous Simulation Theorem \cite{PanduranganPS20} (see also Lemma 2.2 in \cite{dufoulon2023message}) that essentially
the same lower bound (up to a $(\log k)$ factor) holds in the synchronous setting as well if one considers algorithms that  run in (at most) polynomial in $k$
number of rounds. Thus for the rest of the proof, we assume a $\Omega(k/\log k)$ lower bound for the \sd\ problem 
for the {\em synchronous} 2-party communication complexity if we consider only polynomial (in $k$) number of rounds. 

\begin{figure}[h]
\includegraphics[clip,trim=70 180 345 30,width=0.9\textwidth]{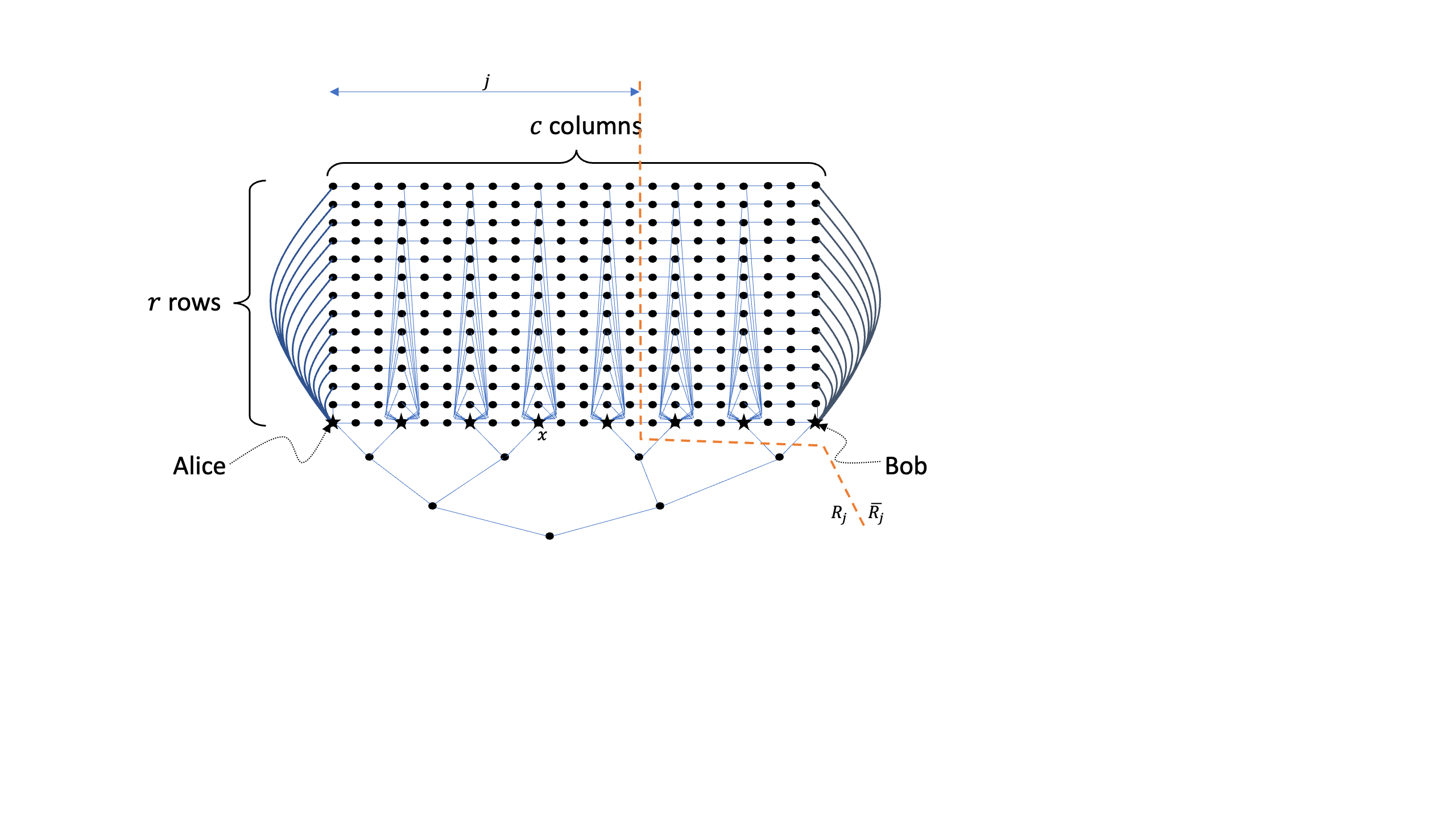}
\caption{\footnotesize Construction of network graph $G_{rc}$ for proving lower bound. The vertices in $X$ are shown as stars (there is a binary tree
at the bottom having the nodes in $X$ as its leaves). One such $x\in X$ is labeled.  The cut induced by an $R_j$ is shown in dotted lines.} 
\label{fig:lb}
\end{figure}

The \dsd\ problem is defined on a graph $G_{rc}$ that is schematically shown in Figure~\ref{fig:lb}. Let $r$ and $c$ be two positive integers such that $rc + \Theta(\log n)= n$ (the network size). We focus on the regime where $c \in \omega(\sqrt{n}\log^3 n)$ and $r \in o(\sqrt{n}/\log^3 n)$. The graph comprises $r$ rows (or parallel paths) $p_\ell$, $1 \le \ell \le r$, with $p_1$ referring to the parallel path at the bottom. Each parallel path comprises $c$ nodes arranged from left to right with the first node referring to the leftmost node and the last node referring to the rightmost node. The first and last nodes in $p_1$ are designated Alice and Bob because they are controlled by the players Alice and Bob in our reduction. Alice (resp., Bob) is connected to the first node (resp., last node) of each $p_\ell$, $2 \le \ell \le r$. Additionally, we pick $\Theta(\log n)$ equally spaced nodes $X$ (of cardinality that is a power of two) from $p_1$ such that the first and last nodes in $p_1$ are included in $X$. For each $x \in X$, say at position $j$ in $p_1$, we add edges from $x$ to the $j$th node in each $p_\ell$, $2 \le \ell \le r$.  Using $X$ as  leaves, we construct a balanced binary tree. We will use $I$ to denote the internal nodes of this tree.   Alice is in possession of bit string $a$ and Bob is in possession of $b$ and, to solve \dsd, they must compute $\disj(a,b)$ in the sleeping model over the network $G_{rc}$. The outputs of all other nodes don't matter; for concreteness, we specify that their outputs should be empty   .

In the \css\ problem defined again on $G_{rc}$, some edges in $G_{rc}$ are marked and at least one node in the network must determine whether the marked edges form a connected spanning subgraph of $G_{rc}$. For the \mst\ problem, we require edges in $G_{rc}$ to be weighted and the goal is to construct a minimum spanning tree of $G_{rc}$ such that the endpoints of each MST edge $e$ are aware that $e$ is an MST edge. Both \css\ and MST must be solved in the sleeping model.

$G_{rc}$ is constructed such that  any node can reach some $x\in X$ within $O(c/\log n)$ steps and any pair of nodes in $X$ are within $O(\log \log n)$ steps (through the tree). Recall that $c \in \omega(\sqrt{n}\log^3 n)$. Thus, we have the following observation.

\begin{observation} \label{obs:exists}
The network graph $G_{rc}$ has diameter $D  \in \Theta(c/\log n)$. Moreover, $D \in \omega(\sqrt{n} \log^* n$). Therefore, \dsd, \css, and \mst\ (if edges in $G_{rc}$ are assigned weights) can be computed in $O(D) = O(c/\log n)$ rounds~\cite{GKP98}. 
\end{observation}

\noindent {\bf Reduction from  \sd\ $\to$ \dsd.}  
\begin{lemma}
Consider an algorithm $P$ in the sleeping model that solves \dsd\  on $G_{rc}$ with $c \in \omega(\sqrt{n}\log^3 n)$ and $r \in o(\sqrt{n}/\log^3 n)$ in $T$ (worst-case) rounds such that $T \in o(c)$ (and we know such an algorithm exists from Observation~\ref{obs:exists}, in particular because $D \in O(c/\log n)$). Then, the  awake time of $P$ must be at least $\Omega(r/\log^2 n)$. This holds even if $P$ is randomized and has an error probability that is bounded by a small constant $\epsilon > 0$.
\label{lem:dsd}
\end{lemma}

\begin{proof}
Suppose for the sake of contradiction $P$ runs in time $T$ and has an awake  complexity of $o(r/\log^2 n)$. Then, we can show that Alice and Bob can simulate $P$ in the classical communication complexity model and solve \sd\ on $r$ bits by exchanging only $o(r/\log n)$ bits which will yield a contradiction to the \sd\ lower bound. We establish this by showing that Alice and Bob can simulate $P$ to solve \sd\ in the classical communication complexity model. 

We show this simulation from Alice's perspective. Bob's perspective will be symmetric. Recall
that $p_\ell$ is the $\ell$th parallel path. Let $p_\ell^j$, $1 \le j \le c$, denote the first $j$ vertices of path~$p_\ell$. We define $R_j$ to be the union of all  $p_\ell^j$ and $I$ (recall that $I$ is the set of the internal nodes of the binary tree), i.e., 
 $   R_j = (\bigcup_{\ell = 1}^r p_\ell^j) \cup I.$ 
Note that $R_j$ induces a cut $(R_j,\bar{R_j})$ that is shown in Figure~1. Alice begins by simulating $R_{c-1}$ in round 1 as she knows the state of all nodes in $R_{c-1}$. At each subsequent round $t$, Alice simulates $R_{c-t}$. Initially, all the information needed for the simulation is available for Alice because the structure of $G_{rc}$ is fully known (except for Bob's input). 

As the simulation progresses, in each round $t > 1$, $t \le T \in  o(c)$, all inputs will be available except for the new bits that may enter $I$ through nodes in $\bar{R}_{c-t}$. Alice will not need to ask for the bits needed by $p_\ell^{c-t}$ because she simulated all nodes in $p_\ell^{c-t+1}$, $1\le \ell \le r$, in the previous round. 
 Note that the portion simulated by Bob will encompass the portion from which Alice may need bits from Bob, so Bob will indeed have the bits requested by Alice. In order to continue the simulation, Alice borrows bits that $P$ transmitted from $\bar{R_{c-t}}$ to $I \cap R_{c-t}$ from Bob. Suppose during the course of the simulation in the communication complexity model, $B$ bits are borrowed from Bob. Then nodes in $I$ must have been awake for a collective total of at least $\Omega(B/\log n)$ rounds (because each message of $O(\log n)$ bits must be received by a node that is awake in $P$).\footnote{Note that all the $B$ bits  cannot solely come through a row path of length $c$, since we are restricting $T\in o(c)$. In other words,  each of the  bits has to go through at least one node in $I$.} 
 This implies that at least one node in $I$ must have been awake for $\Omega(B/\log^2 n)$ rounds because $|I| \in O(\log n)$ and the number of edges incident to nodes in $I$ is also $O(\log n)$ (since nodes in $I$ are of constant degree). 

Since the node awake time is $o(r/\log^2 n)$ for $P$, $B$ must be $o(r/\log n)$, accounting for the fact that each node awake time can potentially transmit $O(\log n)$ bits. But this contradicts the fact that \sd\ requires $\Omega(r/\log n)$ bits in the synchronous communication complexity model.
\end{proof}

\noindent {\bf Reduction from  \dsd\ $\to$ \css.} 

We now show a reduction from \dsd\ $\to$ \css\ by encoding a given $\dsd$ problem instance as a $\css$ instance in the following manner. Recall that in \dsd, Alice and Bob have bit strings $a$ and $b$, respectively, of length $r$ each. Furthermore, recall that Alice (resp., Bob) is connected to first node (resp., last node) of each $p_\ell$, $2 \le \ell \le r$. 

\begin{lemma}
Suppose there is a protocol $Q$ in the sleeping model that solves \css\  on $G_{rc}$ with $c \in \omega(\sqrt{n}\log^3 n)$ and $r \in o(\sqrt{n}/\log^3 n)$ in $T$ rounds such that $T \in o(c)$. Then, the node awake time of $Q$ must be at least $\Omega(r/\log^2 n)$. This holds even if $Q$ is randomized and has an error probability that is bounded by a small constant $\epsilon > 0$.
\label{lem:css1}
\end{lemma}

\noindent {\bf Reduction from \css\ $\to$ \mst.} Recall that \css\ is a decision problem that requires a subset of the edges in the network graph $G_{rc}$ to be marked; we are to report whether the marked edges form a spanning subgraph of $G_{rc}$. \mst\ on the other hand is a construction problem. It takes a weighted network graph and computes the minimum spanning tree. A reduction from \css\ to \mst\ can be constructed by assigning a weight of 1 for marked edges in the \css\ instance and $n$ for all other edges and asking if any edge of weight $n$ is included in the MST. This leads us to the following lemma.

\begin{lemma}
Suppose there is a protocol $M$ in the sleeping model that solves \mst\  on $G_{rc}$ with $c \in \omega(\sqrt{n}\log^3 n)$ and $r \in o(\sqrt{n}/\log^2 n)$ in $T$ rounds such that $T \in o(c)$. Then, the node awake time of $M$ must be at least $\Omega(r/\log^2 n)$. This holds even if $M$ is randomized and has an error probability that is bounded by a small constant $\epsilon > 0$.
\label{lem:css}
\end{lemma}

Thus, we can conclude with the following theorem.

\begin{theorem} 
Consider positive integers $r$ and $c$ such that $rc + \Theta(\log n) = n$ (the network size) and $c \in \omega(\sqrt{n}\log^3 n)$. (Thus, $c$ can range between $\tilde{\Omega}(\sqrt{n})$ to $\tilde{O}(n)$ with $r \in \Theta(n/c)$.) 
Then there exists graphs with diameter at least $\omega(\sqrt{n}\log^3 n)$ such that
any randomized algorithm  for \mst\ in the sleeping model that runs in time $T \in o(c)$ rounds and guaranteed to compute the MST with probability at least $1 - \epsilon$ for  any small fixed $\epsilon > 0$ has worst case awake complexity at least $\Omega(r/\log^2 n)$.  
\label{thm:lowerbound}
\end{theorem}
As an illustrative example of this trade off, let $c = n^{3/4}$ and $r \in \Theta(n^{1/4})$. Consider an MST algorithm that takes $o(n^{3/4})$ rounds; Observation~\ref{obs:exists} guarantees  the existence of such an algorithm. Then, Theorem~\ref{thm:lowerbound} implies that its awake complexity is at least $\tilde{\Omega}(n^{1/4})$.

%% file: algorithms.tex
In this section, we present our algorithms to construct an MST that take optimal awake  complexity.  We first present a toolbox of procedures which are used repeatedly in the subsequent algorithms.  We then develop a randomized algorithm that creates an MST in optimal awake time. We then show how to construct a deterministic algorithm that also is optimal in awake time, however, we pay a cost in terms of slower run time. 

\subsection{Main Ideas \& Useful Procedures}\label{sec:prelims}
\input{prelims}

\subsection{Awake-Optimal Randomized Algorithm}\label{sec:mst-optimal-awake-time}
\input{mst-optimal-awake-time}

\subsection{Awake-Optimal Deterministic Algorithm}\label{sec:det-mst-opt-awake-time}
\input{mst-deterministic}

%% file: prelims.tex
Both of the algorithms we develop in this section can be seen as variations of the classic GHS algorithm to find the MST, adapted to optimize the awake time of the algorithm. Recall that each phase of the GHS algorithm consists of two steps. Step (i) corresponds to finding the minimum outgoing edges (MOEs) for the current fragments and step (ii) involves merging these fragments.

Our algorithms work in phases where at the end of each phase, we ensure that the original graph has been partitioned into a forest of node-disjoint trees that satisfy the following property. For each such tree, all nodes within the tree know the ID of the root of the tree (called \textit{fragment ID}), 
the IDs of their parents and children in the tree, if any, and their {\em distance from the root} (note that it is the hop distance, ignoring the weights) of that tree. We call each such tree a \textit{Labeled Distance Tree (LDT)} and a forest of such trees a \textit{Forest of Labeled Distance Trees (FLDT)}. By the end of the algorithms we design, our goal is to have the FLDT reduced to just one LDT which corresponds to the MST of the original graph. The challenge is to construct an LDT (which will also be an MST) in an awake-optimal manner.

The purpose of maintaining such a structure is that we know how to design fast awake procedures to propagate information within an LDT, which we describe below. Specifically, we know how to set up a schedule of rounds for nodes to wake up in such that (i) information can be passed from the root to all nodes in an LDT in $O(1)$ awake rounds, (ii) information can be passed from a child to the root in $O(1)$ awake rounds, and 
(iii) information can be spread from one LDT to its neighbors in $O(1)$ awake rounds. There are several well-known procedures typically associated with GHS such as broadcast, upcast-min, etc.\ 
that make use of such information propagation. Once we have an LDT, it is easy to implement these
procedures in constant awake time. 
For the processes described below, it is assumed that the initial graph has already been divided into an FLDT where each node $u$ knows the ID of the root, $root$, of the tree it belongs to (i.e., the fragment ID of the tree), the IDs of $u$'s parent and children in that tree, and $u$'s distance to $root$. First of all, we define a transmission schedule that is used in each of the procedures and will be directly utilized in the algorithms. Then we briefly mention the procedures typically associated with GHS.

\input{appendix-prelims}

%% file: appendix-prelims.tex
\paragraph{Transmission schedule of nodes in a tree.}
\begin{sloppypar}
We describe the function $\TS(root, u, n)$, which takes a given node $u$ and a tree rooted at $root$ (to which $u$ belongs) and maps them to a set of rounds in a block of $2n+1$ rounds. $\TS(root,u,n)$ may be used by node $u$ to determine which of the $2n+1$ rounds to be awake in.  
Consider a tree rooted at $root$ and a node $u$ in that tree at distance $i$ from the root. For ease of explanation, we assign names to each of these rounds as well. For all non-root nodes $u$, the set of rounds that $\TS(root,u,n)$ maps to includes rounds $i,i+1, n+1, 2n-i+1$, and $2n-i+2$ with corresponding names $\DOWNRECEIVE, \DOWNSEND, \SIDESENDRECEIVE, \UPRECEIVE$, and $\UPSEND$, respectively. $\TS(root,root,n)$ only maps to the set containing rounds $1$, $n+1$, and $2n+1$ with names $\DOWNSEND$, $\SIDESENDRECEIVE$, and $\UPRECEIVE$, respectively.\footnote{In this description, we assumed that $\TS(\cdot, \cdot, n)$ was started in round $1$. However, if $\TS(\cdot, \cdot, n)$ is started in round $r$, then just add $r-1$ to the values mentioned here and in the previous sentence.}
\end{sloppypar}

\paragraph{Broadcasting a message in a tree.}
We describe the procedure $\DOWNCAST(n)$ that may be run by all nodes in a tree in order for a message $msg$ to be broadcast from the root of the tree to all nodes in the tree. Consider a block of $2n+1$ rounds and let all nodes $u$ in the tree rooted at $root$ utilize $\TS(root,u,n)$ to decide when to be awake. For each $u$, in the round that corresponds to $\DOWNRECEIVE$, $u$ listens for a new message from its parent. In the round that corresponds to $\DOWNSEND$, if $u$ transmits the message it has to its children in the tree.

\begin{observation}\label{obs:downcast}
Procedure $\DOWNCAST(n)$, run by all nodes in a tree, allows the root of that tree to transmit a message to each node in its tree in $O(n)$ running time and $O(1)$ awake time.
\end{observation}

\paragraph{Upcasting the minimum value in a tree.}
We describe the procedure $\UPCASTMIN(n)$ that may be run by all nodes in a   tree rooted at $root$ in order to propagate the minimum among all values held by the nodes of that tree up to $root$. Consider a block of $2n+1$ rounds and let all nodes $u$ in the tree rooted at $root$ utilize $\TS(root,u,n)$ to decide when to be awake. If $u$ is a leaf, then in the round that corresponds to $\UPSEND$, $u$ transmits its message to its parent. For each $u$ that is not a leaf, in the round that corresponds to $\UPRECEIVE$, $u$ listens for messages from its children. In the round that corresponds to $\UPSEND$, $u$ compares the messages it previously received in its $\UPRECEIVE$ round to its current message, if any, and stores the minimum value. It then transmits this minimum value to its parent in the tree.

\begin{observation}\label{obs:upcastmin}
Procedure $\UPCASTMIN(n)$, run by all nodes in a tree, allows the smallest value among all values held by the nodes, if any, to be propagated to the root of the tree in $O(n)$ running time and $O(1)$ awake time.
\end{observation}

\paragraph{Transmitting a message between neighboring nodes in a given tree.}
We describe the procedure $\TRANSMITNEIGHBOR(n)$ that may be run by all nodes in a tree rooted at $root$ so that each node in a tree may transmit a message to its parent and children in the tree. Consider a block of $2n+1$ rounds and let all nodes $u$ in the tree rooted at $root$ utilize $\TS(root,u,n)$ to decide when to be awake. Node $u$ transmits its message in the rounds corresponding to $\DOWNSEND$ and $\UPSEND$. In rounds $\DOWNRECEIVE$ and $\UPRECEIVE$, $u$ listens for messages from its parent and children.

\begin{observation}\label{obs:trasnmit-neighbor}
Procedure $\TRANSMITNEIGHBOR(n)$, run by all nodes in a tree, allows each node in the tree to transmit a message, if any, to its parent and children in $O(n)$ running time and $O(1)$ awake time.
\end{observation}

\paragraph{Transmitting a message between nodes of adjacent  trees.}
We describe the procedure $\TRANSMITADJACENT(n)$ that may be run by all nodes in a   tree rooted at $root$ in order to transmit a message held by nodes in the tree  to their neighbors in adjacent  trees.\footnote{We assume that all nodes in the graph belong to some  tree and are running the algorithm in parallel. We ensure this assumption is valid whenever we call this procedure.} Consider a block of $2n+1$ rounds and let each node $u$ in the tree rooted at $root$ utilize $\TS(root,u,n)$ to decide when to be awake. In the round that corresponds to $\SIDESENDRECEIVE$, $u$ transmits its message to its neighbors not in the tree. It also listens for messages from those neighbors.

\begin{observation}\label{obs:transmit-adjacent}
Procedure $\TRANSMITADJACENT(n)$, run by all nodes in a tree, allows each node in a tree to transfer a message, if any, to neighboring nodes belonging to other trees in $O(n)$ running time and $O(1)$ awake time.
\end{observation}

In the course of our algorithms, we ensure that nodes stay synchronized, i.e., time can be viewed in blocks of $2n+1$ rounds such that all nodes start their first schedule at the same time (and end them at the same time) and continue to start (and end) future schedules at the same time.

%% file: mst-optimal-awake-time.tex
\textbf{Brief Overview.} We describe Algorithm~$\RANDMST$. 
We use a synchronous variant of the classic GHS algorithm (see e.g., \cite{eatcs,peleg}) to find the MST of the graph, where we modify certain parts of the algorithm in order to ensure that the awake time is optimal. Each phase of the classic GHS algorithm  consists of two steps. Step (i) corresponds to finding the minimum outgoing edges (MOEs) for the current fragments and step (ii) involves merging these fragments. 
In the current algorithm, we introduce a new step between those two steps. Specifically, we utilize randomness  to restrict which MOEs are considered ``valid'' for the current phase by having each fragment leader flip a coin and only considering MOEs from fragments whose leaders flip tails to those that flip heads. This restriction ensures that the diameter of any subgraph formed by fragments and valid MOEs is {\em constant}. This restriction, coupled with careful analysis that includes showing that the number of fragments decreases by a constant factor in each phase on expectation, guarantees that the MST of the original graph is obtained after $O(\log n)$ phases with high probability.

\noindent \textbf{Technical Challenges.} As mentioned above, one of the key changes we make is to restrict the MOEs to a subset of ``valid'' ones. This is to address a key technical challenge. When we merge two fragments together, one of those fragments must internally re-orient itself and update its internal values (including distance to the root). This re-alignment and updation takes $O(1)$ awake time. If we have a chain of fragments, say of diameter $d$, we may have to perform this re-alignment procedure $d-1$ times since the re-alignment of fragments is sequential in nature.\footnote{To observe the sequential nature of the re-alignment, consider a chain with three fragments, say $A \leftarrow B \leftarrow C$. Suppose $A$ maintains its orientation. The nodes in $B$ must be processed first and must update their distance to $A$. Only then can the nodes of $C$ accurately update their distance to $A$ (after the node $u$ in $C$ connected to the node $v$ in $B$ learns $v$'s updated distance to $A$).} As a result, if we do not control the length of chains of connected components formed by the fragments and their MOEs, we risk blowing up the awake time of the algorithm. We use randomness to ensure the diameter of any such connected component is a constant. 

As a result of the above change, we have a second technical challenge. Because we reduce the number of valid MOEs, we have to be careful to argue that a sufficient number of fragments are merged together in each phase so that after $O(\log n)$ phases, we end up with exactly one fragment with high probability. We provide such a careful argument, which is somewhat different from the usual argument used in GHS style algorithms.

\noindent \textbf{Detailed Algorithm.}
Algorithm~$\RANDMST$ consists of nodes participating in $4 \lceil \log_{4/3} n \rceil + 1$ phases of the following computations. Recall that between phases we want to maintain an FLDT that is eventually converted into a single LDT. Also recall that in each phase, there are three steps. We describe the three steps in detail below.

{\bf Step (i): Finding MOE of each fragment.} In step (i), all fragments perform the following series of actions. Each node in a fragment participates in the following sequence of procedures: 
first $\DOWNCAST(n)$ to transmit the message to find the  MOE  to all nodes in the fragment (i.e., to find the ``local'' MOEs from each node in the fragment to outside), second $\UPCASTMIN(n)$ to convergecast the smallest among the (local) MOEs  to the root of the fragment to find the overall MOE of the fragment, then $\DOWNCAST(n)$ to transmit the  MOE to all nodes in the fragment, and finally $\TRANSMITADJACENT(n)$ to inform adjacent fragments of the current fragment's MOE. 

\begin{sloppypar}
We note that because of the properties of the LDT, $\DOWNCAST(n)$, $\UPCASTMIN(n)$ and $\TRANSMITADJACENT(n)$ can be accomplished in $O(1)$ awake time and hence the finding MOE of each fragment can be accomplished in $O(1)$ awake time.
\end{sloppypar}

{\bf Step (ii): Finding ``valid'' MOEs.} 
In step (ii), we have each fragment's root flip a coin and only allow MOEs from fragments whose roots flipped tails to those that flipped heads.\footnote{Intuitively, stars are formed by fragments and MOEs as a result of this process. In each such star, the center fragment is one whose root flipped heads and the leaves, if any, are fragments whose roots flipped tails.} Henceforth, if a fragment's root flipped a heads (tails), then we say that the fragment flipped a heads (tails). Alternatively, we can say that the fragment is considered a heads (tails) fragment. Each fragment root flips an unbiased coin and uses $\DOWNCAST(n)$ to inform the nodes in the fragment if the root flipped a tails or a heads. All nodes participate in a $\TRANSMITADJACENT(n)$ to inform adjacent fragments of the current fragment's coin flip. Now, we only consider an MOE as ``valid'' if it originated from a fragment that flipped a tails and is to a fragment that flipped a heads. Now, each node adjacent to an MOE knows if that MOE is a valid one or not. By using $\UPCASTMIN(n)$ and $\DOWNCAST(n)$ to transmit the information of ``valid'' or ``invalid'', all nodes in the fragment will know if the outgoing MOE is valid or not.

{\bf Step (iii): Merging fragments.} In step (iii), we merge each subgraph formed by fragments and valid MOEs into a single fragment. Consider a subgraph consisting of several tails fragments whose MOEs lead to a single heads fragment. The heads fragment retains its fragment ID while the remaining fragments take on the ID of the heads fragment. Furthermore, these other fragments also re-orient themselves such that they form subtrees of the heads fragment. Specifically, consider a tails fragment $T$ with root $root_T$ and an MOE to a heads fragment $H$ where nodes $u_T$ and $u_H$ are the nodes of the MOE belonging to $T$ and $H$, respectively. The nodes in fragment $T$ re-orient themselves such that $u_T$ is the new root of the tree. Additionally, $u_T$ considers $u_H$ its parent in the merged graph. 
This is described below in Procedure~$\REORIENTFRAG(n)$, a process similar to that in~\cite{BM21}. The process is also illustrated in Figures~\ref{fig:reorient-tree1},~\ref{fig:reorient-tree2},~\ref{fig:reorient-tree3}, and~\ref{fig:reorient-tree4}. 
Recall that at the end of the previous step, each node in the fragment knows whether the fragment root flipped heads or tails and whether the MOE leading out of the fragment is a valid MOE or not. 

\input{appendix-figures}

As mentioned earlier in Section~\ref{sec:algorithms}, each node in an  
LDT $T$ maintains the ID of the fragment it belongs to as well as the distance of the node from the root of the fragment. We show how to correctly maintain this information when we merge LDTs with the help of two temporary variables, $\NEWFRAGID$ and $\NEWLEVELNUM$, respectively. Additionally, each node of a fragment that updates these values must also re-orient itself by updating its parent and child pointers so that the overall merged fragment is also a tree. In the course of describing how to update fragment ID and distance to the root, we also show how to identify a node's new parent and children, if any. This information is recorded by the node and we mention when the node internally updates this information.\footnote{If the re-orientation information is updated immediately, complications may arise when we call a subprocedure that is to be run using values from the old tree. To avoid this, the original values are temporarily retained and updation is performed later on.}

We now formally describe Procedure~$\REORIENTFRAG(n)$. First, all nodes participate in $\TRANSMITADJACENT(n)$ to transmit their fragment ID and level number. Now $u_T$ sets its $\NEWLEVELNUM$ to that of $u_H$ plus one and stores the fragment ID of $u_H$ in $\NEWFRAGID$. Additionally, $u_T$ records the info that $u_H$ is its new parent and its neighbors in $T$ are its new children, to be updated internally later on. The remaining nodes of $T$ initialize $\NEWFRAGID$ and $\NEWLEVELNUM$ to $\bot$, which we also refer to as those values being \textit{empty}.  
Now, all nodes $v$ in the fragment participate in two instances of $\TS(root_T,v,n)$. At a high level, the first instance is used to update the $\NEWFRAGID$ and and $\NEWLEVELNUM$ of nodes on the path from $u_T$ to $root_T$. The second instance is used to update the required variables for all the remaining nodes.

In the first instance, each node $v$ during its $\UPSEND$ round sends up the value in its $\NEWLEVELNUM$ if it is not empty 
and $v$ also sends up the value of $\NEWFRAGID$ if it is not empty. 
During an $\UPRECEIVE$ round, if $v$ receives a non-empty $\NEWLEVELNUM$ from its child, $v$ sets its own $\NEWLEVELNUM$ to the received value plus one. Similarly, if $v$ receives a non-empty $\NEWFRAGID$, $v$ sets its own $\NEWFRAGID$ to that value. Additionally, if $v$ receives a non-empty $\NEWLEVELNUM$, from its child, it records internally that its child will be its new parent and its neighbors in $T$ will be its children.

In the second instance of $\TS(root_T,v,n)$, each node $v$ during its $\DOWNSEND$ round sends down the value of its $\NEWLEVELNUM$ and also the value of its $\NEWFRAGID$. During a $\DOWNRECEIVE$ round, if $v$'s $\NEWLEVELNUM$ is non-empty and it receives a non-empty value from its parent, $v$ updates its $\NEWLEVELNUM$ to that value plus one. Similarly, if $v$'s $\NEWFRAGID$ is non-empty and it receives a non-empty value from its parent, $v$ updates its $\NEWFRAGID$ to that value. 

At the end of step (ii), each node $v$ updates its fragment ID to $\NEWFRAGID$ and updates its level number to $\NEWLEVELNUM$, assuming they are non-empty, and subsequently clears those variables, i.e., sets their values to $\bot$, in preparation for the next phase.\footnote{$\NEWFRAGID$ and $\NEWLEVELNUM$ are variables that serve a temporary purpose each phase. As such, we do not need to retain their values across phases.} 
Additionally, any information about re-orientation, i.e., updating parent and/or children, is now updated locally within the tree. 

\textbf{Analysis.} We now prove that the algorithm correctly outputs the MST of the original graph with the desired running time and awake time.

Recall that the number of fragments can never increase from one phase to the next. Let phase $\IMPPHASE$ correspond to the last phase in which there is more than one fragment at the beginning of the phase. 
We will show that $\IMPPHASE = 4 \lceil \log_{4/3} n \rceil$. The following lemma shows that for the first $\IMPPHASE$ phases of the algorithm, the number of fragments is reduced by a constant factor in each phase with high probability. 

\begin{lemma}\label{lem:frag-reduction}
For each phase of Algorithm $\RANDMST$ where there are initially at least two fragments at the start of that phase, the number of fragments is reduced by at least a factor of $4/3$ in that phase on expectation. Furthermore, by phase $4 \lceil \log_{4/3} n \rceil+1$, there is at most one fragment in the graph.
\end{lemma}

\begin{proof}
Denote by $\mathcal{F}_i$ the set of fragments that are present at the beginning of phase $i \in [1,\IMPPHASE+1]$. Define $F_i = |\mathcal{F}_i|$. We want to show that for all $i \in [1, \IMPPHASE]$, $E[F_{i+1}] = 3E[F_i]/4$. We can then leverage this to show a high probability bound on the number of fragments present in the system at the end of phase $\IMPPHASE$. It is also easy to see that we only need to prove that the current lemma holds until phase $\IMPPHASE$ ends as by definition phase $\IMPPHASE$ is the last phase in which there are two or more fragments.

Consider some arbitrary phase $i$, $i \in [1, \IMPPHASE]$. At the beginning of phase $i$, there are some $F_i$ fragments present. From the law of total expectation (since $F_i$ and $F_{i+1}$) are random variables from the same probability space), we know that $E[F_{i+1}] = E[E[F_{i+1}|F_i]]$. Thus, we now upper bound $E[F_{i+1}|F_i]$. 

In order to count the number of fragments that are left after the end of a phase, we first define the notion of a given fragment ``surviving'' a given phase. Consider one of the $F_i$ fragments at the beginning of phase $i$, call it $f$. Let $f$'s MOE lead to another fragment $f'$. We define $f$ to survive the phase if either $f$ flipped heads or if both $f$ and $f'$ flipped tails. In other words, $f$ survives the phase if it does not flip tails and merge into a heads fragment. Thus, the probability that $f$ survives a given phase is $\leq 3/4$. Let $X_f^i$ be an indicator random variable that takes value $1$ if fragment $f$ survives phase $i$ and $0$ otherwise. It is easy to see that the expected number of fragments that survive phase $i$ acts as an upper bound on $E[F_{i+1}|F_i]$. 

Now, we can calculate the value of $E[F_{i+1}]$ as follows:
$$
    E[F_{i+1}] = E[E[F_{i+1}|F_i]] 
    \leq E\left[\sum\limits_{f \in \mathcal{F}_i} X_f^i\right] 
    \leq \sum\limits_{f \in \mathcal{F}_i} E\left[X_f^i\right] 
    \leq \sum\limits_{f \in \mathcal{F}_i}  Pr(f \text{ survives phase } i) 
    \leq 3 F_i/4.
$$

Recall that the number of fragments at the beginning of phase $1$ is equal to $n$. Thus, we have $E[F_{\IMPPHASE+1}] \leq E[F_{\IMPPHASE}]/(4/3) \leq E[F_1]/(4/3)^{\IMPPHASE} \leq n/ (4/3)^{\IMPPHASE}$.  So we see that after $\IMPPHASE$ phases, we have the following by Markov's inequality:
$$     Pr(F_{\IMPPHASE+1} > 1) \leq E[F_{\IMPPHASE+1}]/1 
    \leq n/ (4/3)^{\IMPPHASE}.
$$

By setting $\IMPPHASE$ to $4 \lceil \log_{4/3} n \rceil$, we see that $Pr(F_{\IMPPHASE+1} > 1) \leq 1/n^3$. Thus with high probability, the number of fragments at the beginning of phase $\IMPPHASE+1$ does not exceed $1$.
\end{proof}

We are now ready to argue that the algorithm is correct.

\begin{lemma}\label{lem:alg-correctness}
Algorithm $\RANDMST$ results in each node of the initial graph knowing which of its edges are in the MST with high probability.
\end{lemma}

\begin{proof}
By Lemma~\ref{lem:frag-reduction}, we see that in phases $i \in [1, \IMPPHASE]$, the number of fragments is reduced by a constant factor until only one is left with high probability. From phase $\IMPPHASE$ onwards, there is at most one fragment with high probability and the number of fragments never increases in a given phase. Thus, since we run the algorithm for $\IMPPHASE+1$ phases and the initial graph is connected, we see that at the end of the algorithm, there exists only one fragment. Furthermore, it is easy to see from the algorithm that each node in the initial graph belongs to this fragment and every edge in this fragment was once the minimum outgoing edge from some fragment to another. In other words, the set of nodes and edges within the final fragment represent the MST of the initial graph.
\end{proof}

We also bound the running time and awake time of the algorithm below.

\begin{lemma}\label{lem:runtime-awaketime}
Algorithm $\RANDMST$ takes $O(n \log n)$ running time and $O(\log n)$ awake time. 
\end{lemma}

\begin{proof}
Recall that there are $O(\log n)$ phases of the algorithm. We show that in each phase, each node experiences awake time of $O(1)$ and running time of $O(n)$, thus giving us the desired bounds. In a given phase, a given node may run the following procedures a constant number of times: $\DOWNCAST(n)$, $\UPCASTMIN(n)$, and $\TRANSMITADJACENT(n)$. Each of these procedures takes $O(1)$ awake time and $O(n)$ running time. Additionally, to complete step (iii) in a given phase (i.e., to run $\REORIENTFRAG(n)$), a node $u$ belonging to a fragment with root $root$ may need to perform one instance of $\TRANSMITADJACENT(n)$ and two instances of $\TS(root,u,n)$, which requires $O(1)$ awake time and $O(n)$ running time.
\end{proof}

\begin{theorem}\label{the:optimalmst-randomized-theorem}
Algorithm $\RANDMST$ is a randomized algorithm to find the MST of a 
graph with high probability in $O(n \log n)$ running time and $O(\log n)$ awake time.
\end{theorem}

%% file: appendix-figures.tex
\begin{figure}[h]
\includegraphics[page=1,width=.8\textwidth]{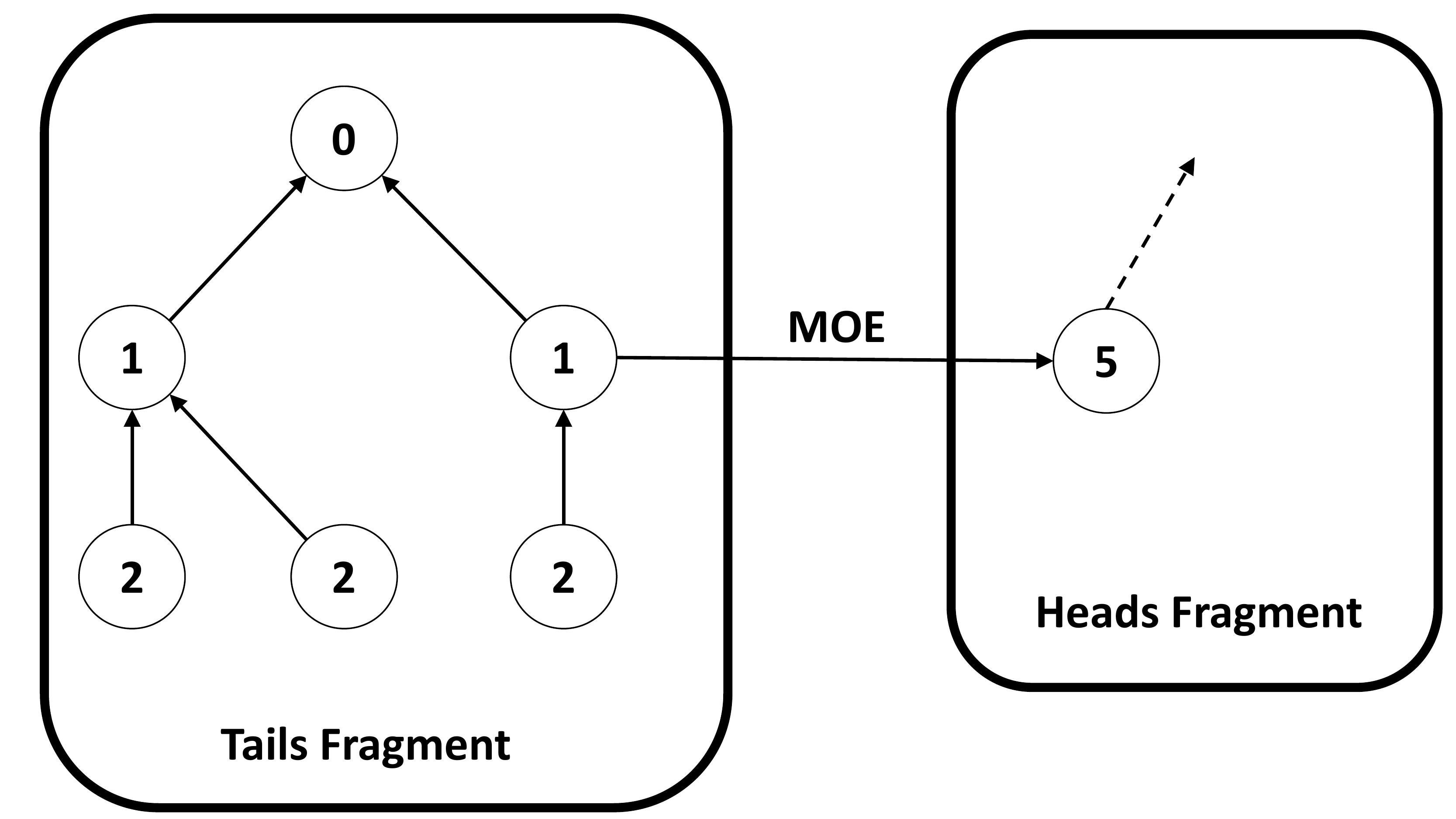}
\caption{Initial configuration. Tails fragment has an MOE to a Heads fragment. Numbers in nodes refer to distance from root of that tree.}\label{fig:reorient-tree1} 
\end{figure}
	
\begin{figure}[h]
\includegraphics[page=2,width=.8\textwidth]{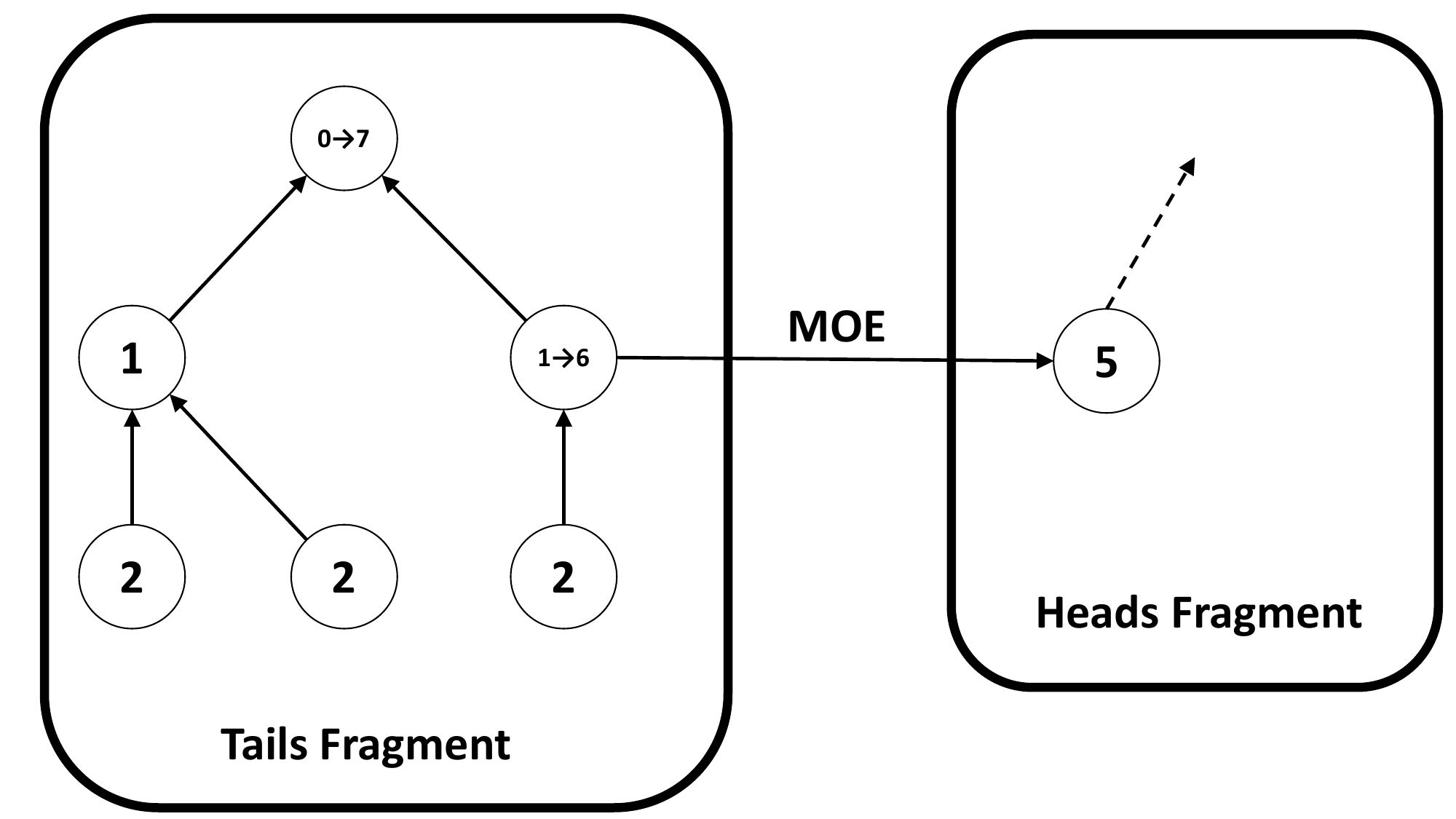}
\caption{After $\TRANSMITADJACENT(n)$ and first instance of $\TS(\cdot,\cdot,n)$, nodes from MOE node to root updated distance from root values internally. Fragment ID of Heads fragment propagated to these nodes as well.}\label{fig:reorient-tree2} 
\end{figure}

\begin{figure}[h]
\includegraphics[page=3,width=.8\textwidth]{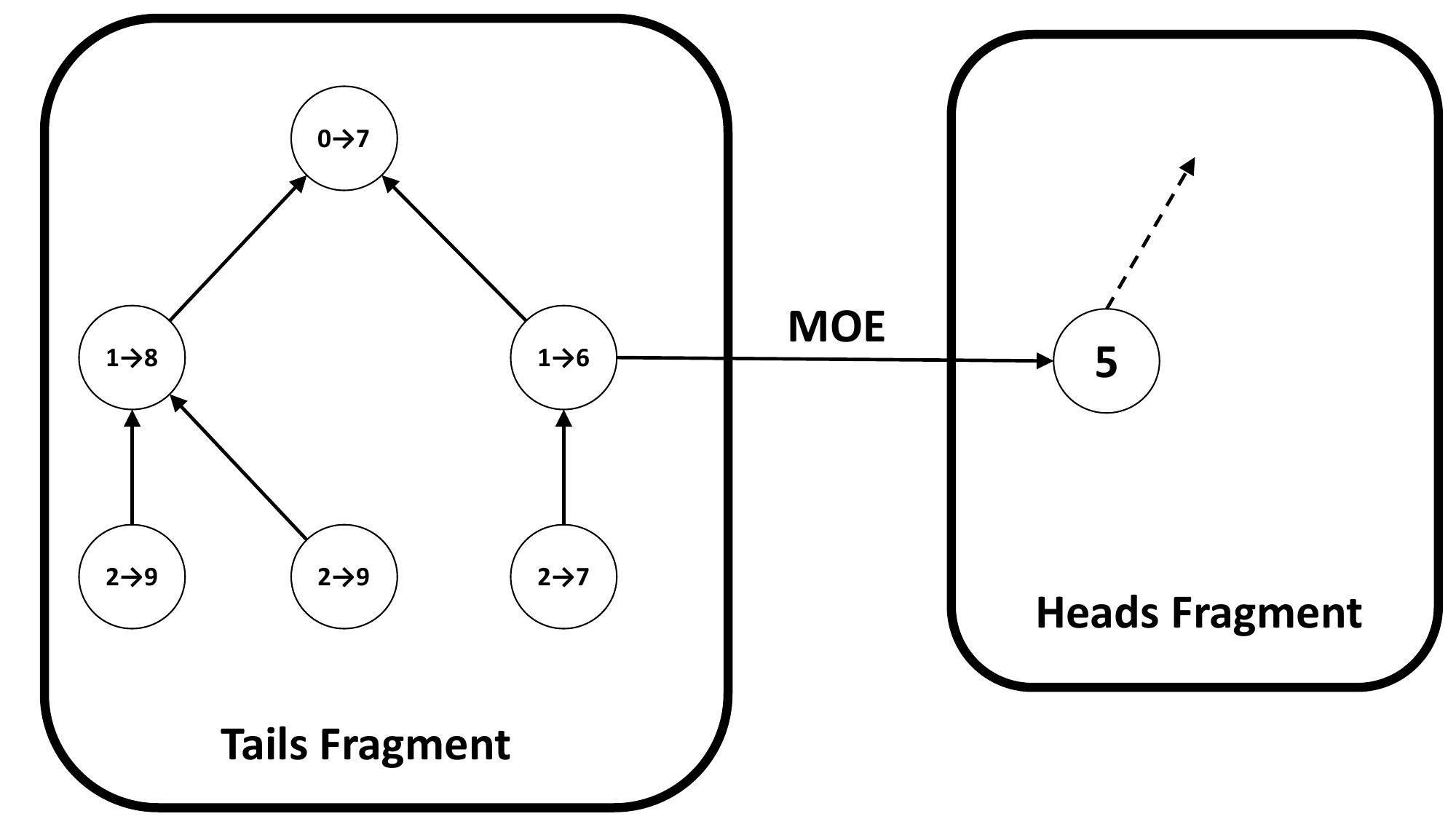}
\caption{After second instance of $\TS(\cdot,\cdot,n)$, remaining nodes updated their distance values internally. Fragment ID of Heads fragment propagated to these nodes as well.}\label{fig:reorient-tree3} 
\end{figure}

\begin{figure}[h]
\includegraphics[page=4,width=.8\textwidth]{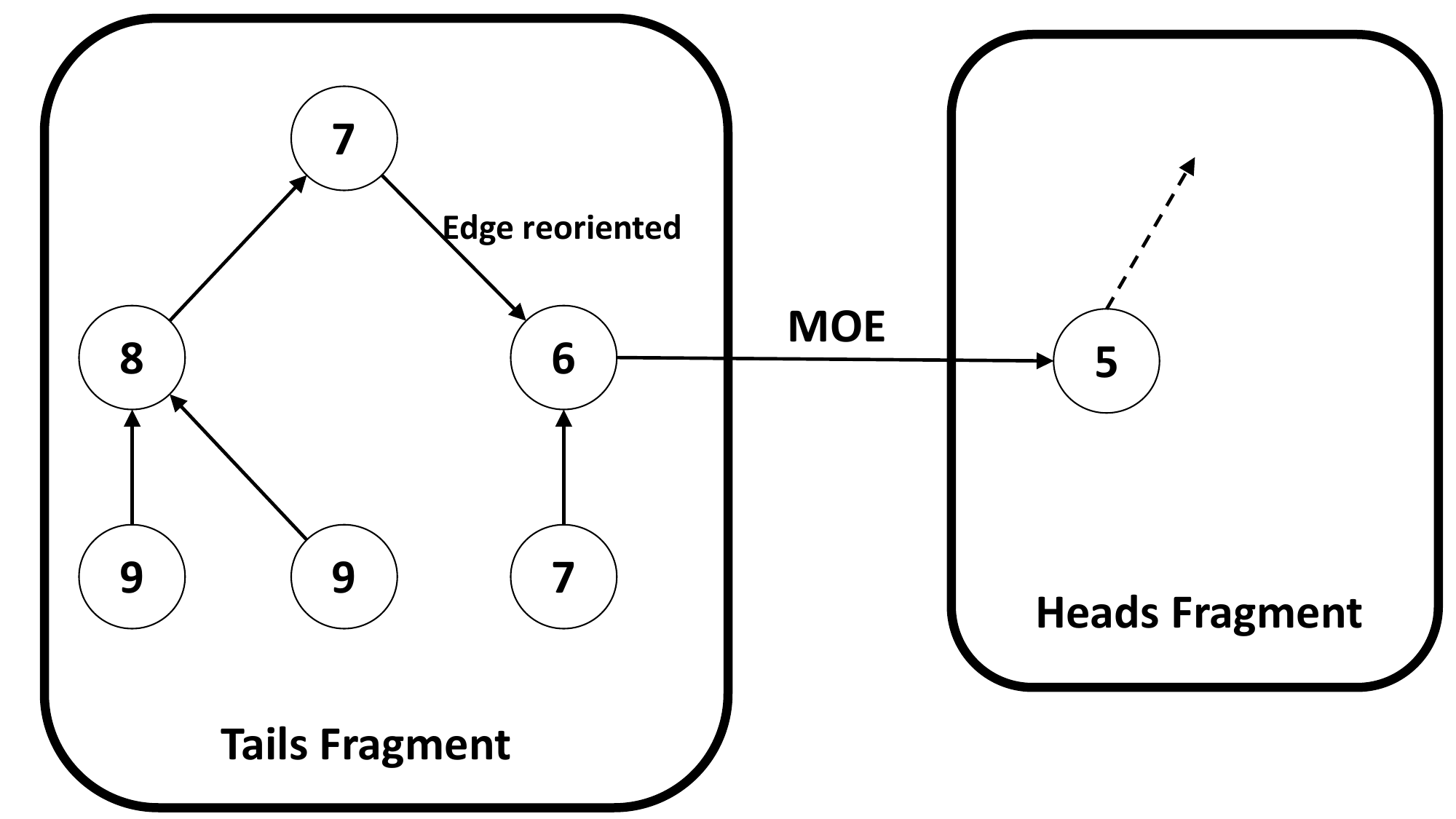}
\caption{Finally, all nodes update their distance from root values (which were stored in temporary variables). Furthermore, edges from MOE node to root are re-oriented to reflect new tree structure.}\label{fig:reorient-tree4} 
\end{figure}

%% file: mst-deterministic.tex

\textbf{Brief Overview.} We describe Algorithm $\DETMST$, which is similar to Algorithm $\RANDMST$ (described in Section~\ref{sec:mst-optimal-awake-time}).  Unlike Algorithm $\RANDMST$, where we use random coin flips in step (ii) to limit the diameter of subgraphs in the \emph{fragment graph}  (where each
fragment is a (super)node and the edges between the fragments  are the  MOEs),
the main technical challenge is to deterministically keep the diameter of these merging subgraphs (in the fragment graph) to be {\em constant};
this is crucial for implementing one phase of merging in $O(1)$ awake time.\footnote{One can accomplish this deterministically
in $O(\log^*n)$ time even in the traditional model using coloring or maximal matching in the fragment graph which is a tree (see e.g., 
\cite{PanduranganRS17}), but this will lead to an overhead of a $O(\log^*n)$ factor in the awake time.}

We use the following combination of techniques. In step (i), we first sparsify such graphs by allowing any fragment to be adjacent to at most $4$ other fragments. Then, in step (ii), we use a fast awake time coloring algorithm and selective merging to ensure that subgraphs do not have a large diameter. We leverage this property to ensure that fast awake time merging of these fragments occurs. We also ensure that a sufficient number of fragments participate in merging so that the number of fragments decreases by a constant factor from one phase to the next. The above modifications result in an algorithm that outputs the MST of the original graph in $O(\log n)$ phases. 
Note that we require nodes to know the value of $N$, the range from which node IDs are taken.

\noindent \textbf{Technical Challenges.} We experience similar technical challenges as those faced when designing Algorithm $\RANDMST$. However, we resolve those issues quite differently here. As before, when we construct connected components of fragments and their MOEs, we want that the diameter of each of these components is a constant, so that we can re-orient fragments quickly during the merging process. Since we do not have access to randomness, we rely on the somewhat standard approach of using a deterministic maximal matching over these components to reduce the diameter. However, while the approach is standard, the execution is not. In order to maintain a small awake time, we first reduce every component to a set of bounded degree components and then run a tailored algorithm to color all fragments in $O(1)$ awake time. In particular, we construct a series of independent sets of the fragment graph over $O(\log^4 n)$ stages, such that each fragment belongs to exactly one independent set. In the corresponding stage, a fragment in the independent set colors itself with the first available color and informs its neighbors. Subsequently, converting this coloring to a maximal matching is done as usual.

Once again, due to the above changes, we must deal with a second technical challenge. Because we reduce the number of valid MOEs, we have to be careful to argue that a sufficient number of fragments are merged together in each phase so that after $O(\log n)$ phases, we end up with exactly one fragment. We utilize an interesting combinatorial approach to argue this.

\textbf{Detailed Algorithm.}
We now give a detailed break up of each phase of the algorithm. Let $\NUMFRAGSDET = 240000$. 
Recall that there are $\lceil \log_{\NUMFRAGSDET/(\NUMFRAGSDET-1)} n \rceil + \NUMFRAGSDET$ phases and in each phase, there are three steps.\footnote{We have not chosen to optimize the constants
in our analysis.} We describe each in detail separately. 

\begin{sloppypar}
 {\bf Step (i): Finding MOE of each fragment.} In step (i), we find each fragment's MOE in the same way as 
 step (i)  of Algorithm $\RANDMST$.
\end{sloppypar}

 {\bf Step (ii): Finding ``valid'' MOEs.} 
In step (ii), each node in a given fragment knows, for each of its edges adjacent to it, whether that edge is an MOE from some other fragment to the given fragment. Let us differentiate these MOEs \textit{to} the fragment from the MOE \textit{from} the fragment by calling the former $\INCOMINGMOE$s. 
Now, we have each fragment select up to $3$ ``valid'' MOEs from its $\INCOMINGMOE$s, chosen arbitrarily. This is in contrast to how the valid MOEs were chosen during Algorithm~$\RANDMST$, where we used coin flips to determine valid MOEs. Define an \textit{incoming MOE node $v$ of fragment $f$} as a node $v$ belonging to fragment $f$ such that $v$ is adjacent to an edge that is an MOE from some other fragment to $f$. In the context of a given fragment $f$, define a \textit{valid MOE child node $v$ of a node $u$} as a child node of $u$ such that in the subtree rooted at $v$, there exists an incoming MOE node of $f$. 
At a high level, the total number of $\INCOMINGMOE$s is communicated to the root of the fragment. The root then allots up to $3$ virtual ``tokens'' (i) to its valid MOE child nodes to be used to select $\INCOMINGMOE$s and (ii) to itself if the root is an incoming MOE node.
Any node that receives one or more such tokens distributes them among its valid MOE child nodes and itself if it is an incoming MOE node. This process is repeated until all tokens are transmitted to incoming MOE nodes of the fragment. 

In more detail, consider a given fragment $f$ with root $root$. Each node $u$ that belongs to $f$ runs $\TS(root,u,n)$. In $u$'s $\UPRECEIVE$ round, $u$ receives info on the number of nodes in its children's subtrees (if any) that are incoming MOE nodes. In $u$'s $\UPSEND$ round, $u$ aggregates the total number of such incoming MOE nodes. It adds $1$ to that number if $u$ is itself an incoming MOE node. Then $u$ sends this number up to its parent in the fragment. Now, at the end of $\TS(root,u,n)$, the root of the fragment is aware of how many nodes in the fragment (including the root itself) are incoming MOE nodes. If that number is $\leq 3$, then all of them are accepted as valid, otherwise at most 3 of them are selected as valid as follows. All nodes $u$ in the fragment participate in $\TS(root,u,n)$. If the root is an incoming MOE node, it selects itself as a valid incoming MOE node. The root then decides in an arbitrary manner how many nodes from each of its children's subtrees can become valid MOE nodes. This information is transmitted during the $root$'s $\DOWNSEND$ round in the form a number indicating how many MOE nodes can be in the subtree rooted at that child. As for an arbitrary node $u$ in the fragment, in $u$'s $\DOWNRECEIVE$ round, it receives this number from its parent. Subsequently, during $u$'s $\DOWNSEND$ round, $u$ does the following. If $u$ is an incoming MOE node, then $u$ selects itself as a valid incoming MOE node and reduces the number by one. If the number $u$ has is non-zero, then $u$ distributes the number among its children that have $\INCOMINGMOE$s in their trees. 

For every fragment $f$, the information of whether an incoming MOE from a fragment $f'$ is valid or not is communicated to each $f'$ as follows. All nodes in $f$ and every other fragment run $\TRANSMITADJACENT(n)$ and MOE nodes communicate  whether they are valid or not to their neighbors. Let $u$ be an incoming MOE node from fragment $f$ whose neighbor in $f'$ (and the other end of the MOE) is node $v$. Now, $v$ from $f'$ knows if $(v,u)$ is a valid MOE from $f'$ to $f$ or not. Now all nodes participate in $\UPCASTMIN$ with the following conditions: (i) nodes other than $v$ send up $\infty$, (ii) if $(v,u)$ is a valid MOE, then $v$ sends up the edge weight of $(v,u)$ else it sends up $-\infty$ 
As a result, once $\UPCASTMIN$ is finished, the root of fragment $f'$ will see that $(v,u)$ is the MOE with either its original weight or $-\infty$ as the edge weight and can thus conclude if $(v,u)$ is a valid MOE or not.

{\bf Step (iii): Merging fragments.} In step (iii), the nodes of each fragment $f$ learn the fragment IDs of fragments up to 2 hops away from $f$ in the fragment graph $G'$ (formed by the fragments as (super)nodes and valid MOEs as edges). Nodes in each fragment $f$ first collect information about valid MOEs (both incoming \& outgoing) at the root of the fragment $f$. Subsequently, all nodes in $f$ participate in a $\DOWNCAST(n)$ to send this information (which takes only $O(\log n)$ bits to encode information about \textit{all} such MOEs). This information about the (at most) four neighbors of a fragment $f$ in $G'$ is then sent to all nodes of neighboring fragments of $f$ in $G'$ via one instance of $\TRANSMITADJACENT(n)$, another instance of gathering information, and subsequent $\DOWNCAST(n)$. Thus, all nodes in $f$ learn of fragments $2$ hops away from $f$ in $G'$. More details are given below.

Recall that at the end of step (ii), each node in  the fragment graph $G'$ has at most $3$ valid incoming MOEs and at most one valid outgoing MOE. Thus, the maximum degree of any node in $G'$  is $4$. 
All nodes in each fragment $f$ run Procedure $\NEIGHBORAWARE(n)$, a variant of $\UPCASTMIN(n)$, where (i) each valid MOE node $u$ (both incoming and outgoing) to some fragment $f'$ sends up a tuple of $\langle$ $u$'s ID, $u$'s fragment ID, weight of $u$'s edge to $f'$, fragment ID of $f'$, color of $f'$ $\rangle$, (ii) other nodes send up a tuple of $\langle \infty, \infty, \infty, \infty, -1 \rangle$ if they do not have another value to send up, and (iii) each node, instead of only sending up the one tuple it knows of, sends up a concatenation of the (at most $4$) tuples that have non-$\infty$ values. Note that the fourth value of the tuple, color of $f'$, is set to $-1$ throughout this procedure as the fragments are as yet uncolored. However, this procedure is used elsewhere where color is needed. Now, all nodes participate in $\DOWNCAST(n)$ so that all nodes of $f$ are aware of tuples corresponding to the at most $4$ neighboring fragments of with valid MOEs from/to $f$. Each node stores this information in the variable $\NBRINFO$. Subsequently, all nodes participate in one instance of $\TRANSMITADJACENT(n)$, where the nodes of fragment $f$ inform neighboring nodes of the at most $4$ valid MOEs. Another instance of $\NEIGHBORAWARE(n)$ is run by all nodes, where each valid MOE node sends up the at most $4$ tuples that it heard about from its neighbor (instead of the one tuple relating to itself), resulting in at most $16$ tuples being concatenated together and sent to the root of $f$. Finally, all nodes participate in $\DOWNCAST(n)$ so that all nodes of $f$ are aware of the at most $16$ neighbors of $f$ in the fragment graph.

Subsequently, we color the fragments and then selectively merge them. Consider a color palette consisting of colors Blue, Red, Orange, Black, and Green. Furthermore, let there exist a total ordering on this palette based on the relation of priority, where we say that a color $A$ has a higher priority than color $B$ and denote the relation by $A>B$, such that Blue $>$ Red $>$ Orange $>$ Black $>$ Green.  
Let us consider the fragment graph $G'$ as computed from step (ii). Recall that the maximum degree of any node in $G'$ is $4$, so $5$ colors are sufficient for coloring (there always exists a $\Delta + 1$ coloring of a graph with maximum degree $\Delta$). At a high level, we assign each fragment $f$ to some independent set of $G'$ (there are at most $O(\log^4 n)$ independent sets formed) and a corresponding stage $i$. In stage $i$, (the nodes of) $f$ wakes up and colors itself with the first available color, i.e., the highest priority color not chosen by any of its neighbors in $G'$. 

We now describe the coloring in more detail. Recall that the ID of each fragment requires $\log N$ bits. Consider all possible choices of $4$ bits  $b_1, b_2, b_3, b_4$ from these $\log N$ bits, i.e., $O(\log^4 N) = O(\log^4 n)$ choices of $4$ bits. For each of these choices, consider all possible $2^4$ assignments of values from $\lbrace 0,1 \rbrace$ to these bits. All nodes can easily deterministically construct the same schedule with $O(\log^4 n)$ entries, where each entry represents one choice of bits and a unique assignment to these chosen bits. To perform the coloring, all nodes will participate in $O(\log^4 n)$ stages, each of $O(n)$ rounds, where each stage corresponds to an entry of this schedule. For a fragment $f$ and some chosen bits $b_1,b_2,b_3,b_4 \in [1, \log N]$ where $b_1 \neq b_2 \neq b_3 \neq b_4$, we say that \textit{a bit assignment $A(b_1),A(b_2),A(b_3),A(b_4) \in \lbrace 0, 1\rbrace$ is valid for $f$} when the corresponding bits of the ID of $f$ have those values. Now, for each fragment $f$, we define its \textit{active stage} as the stage where three conditions are satisfied: (i) $f$ is uncolored, (ii) $f$'s bit assignment is valid, and (iii) none of $f$'s neighbors' (in $G'$) bit assignments are valid. Notice that each fragment has exactly one active stage. A fragment wakes up in its active stage and colors itself with the highest priority color not chosen by any of its neighbors in $G'$, and informs its neighboring fragments of the color it chose. Each fragment wakes up in its active stage and the active stage of each of its (at most four) neighboring fragments (for a given fragment $f$ and its neighbor $f'$, since $f$ knows the IDs of $f'$'s neighbors in $G'$, $f$ knows the active stage of $f'$). Now, we formally describe the Procedure $\AWAKECOLORING(n,N)$ used to achieve this coloring.

All nodes perform the Procedure $\AWAKECOLORING(n,N)$, which requires each node to know both the total number of nodes $n$ and the range from which every ID is taken $[1,N]$. Recall that all nodes in each fragment $f$ are aware of the IDs of neighboring fragments in the fragment graph $G'$ up to two hops away from $f$. All nodes participate in the following $O(\log^4 n)$ stage process, where each stage consists of sufficient rounds to perform one instance each of $\TRANSMITADJACENT(n)$, $\NEIGHBORAWARE(n)$, and $\DOWNCAST(n)$. In stage $i$, active fragments and neighbors of active fragments wake up. If fragment $f$ is active, the nodes of $f$  choose the highest priority color not chosen by any of $f$'s neighboring fragments. Subsequently, all awake nodes  participate in one instance of $\TRANSMITADJACENT(n)$ to learn of neighboring fragments' colors. This information is disseminated to all nodes in a given awake fragment via one instance of $\NEIGHBORAWARE(n)$ and $\DOWNCAST(n)$.

Now, we describe the selective merging in more detail. We first identify the set of fragments that will merge into other fragments and will thus ``not survive'' the phase. These are all fragments that were colored Blue during Procedure~$\AWAKECOLORING(n,N)$. Recall that there are two types of fragments that are colored Blue. Those with neighbors in $G'$ and those without, which we call \textit{singleton fragments}. 

Those Blue fragments with neighbors pick one of their neighbors in $G'$ arbitrarily (which is of course a non-Blue fragment) and then merge into them. This can be achieved by running Procedure~$\REORIENTFRAG(n)$ from Section~\ref{sec:mst-optimal-awake-time} where 
we consider Blue fragments as Tails fragments and all non-Blue fragments as Heads fragments.
The merged fragment takes on the fragment ID of the fragment that acted as the Heads fragment.

\begin{sloppypar}
We now look at how to merge the singleton fragments into the remaining fragments. During Procedure~$\AWAKECOLORING(n,N)$, these singleton fragments are colored Blue but they did not merge into other fragments during the previous instance of Procedure~$\REORIENTFRAG(n)$. We have each of these Blue fragments merge into the fragment at the end of its MOE. But before doing that, we need each node in these Blue fragments to become aware of any changes to fragment IDs and level numbers of neighboring nodes. This information is conveyed by having all nodes in the original graph participate in one instance of $\TRANSMITADJACENT(n)$ to inform nodes in neighboring fragments about such updates.  
Subsequently, all nodes participate in Procedure~$\REORIENTFRAG(n)$ where we consider these Blue singleton fragments as Tails fragments and all remaining fragments as Heads fragments. The merged fragment takes on the fragment ID of the fragment that acted as the Heads fragment. At the end of all this, each of the previous singleton fragments is merged into some other fragment and does not survive the phase.
\end{sloppypar}

\textbf{Analysis.} We now prove that the algorithm correctly outputs the MST of the original graph with the desired running time and awake time. 

We first show that Procedure~$\AWAKECOLORING(n,N)$ results in all nodes being colored correctly and each node only participating in at most $5$ stages.

\begin{lemma}\label{lem:awake-coloring-works}
All fragments are correctly colored as a result of Procedure~$\AWAKECOLORING(n,N)$ and each fragment is awake in at most $5$ stages.
\end{lemma}

\begin{proof}
We first show that each fragment is active in exactly one stage. Recall that a fragment $f$ is active if it is (i) uncolored, (ii) the chosen bits and bit assignment are valid for $f$, and (iii) the bit assignments for $f$'s neighbors are not valid. It is easy to see that $f$ will be active in at most one stage, as in the first stage it is active, it will become colored, and so will not become active again. So, we need to show that it is active in at least one stage. Initially, $f$ is uncolored and so condition (i) is satisfied until $f$ reaches an active stage. Since all fragments have unique IDs, there exists at least one bit value that differs between $f$ and each of its neighbors in $G'$. In some stage, these bits will be the chosen bits with the bit assignment corresponding to the values of $f$'s ID. Thus, there exists a stage where $f$'s bit assignment is valid and all of $f$'s neighbors' bit assignments will not be valid. 

It is easy to see that each fragment will be awake in at most $5$ stages since it will be active (and thus awake) in one stage, and awake for each of its (at most four) neighbors' active stages.

To argue that the coloring is correct, notice that the set of fragments that are active in the same stage is an independent set (of the supergraph $G'$). This is because conditions (ii) and (iii) preclude both a fragment and its neighbor from being active in the same stage. Thus, when a fragment is active, any color it chooses will not clash with its neighbors' colors (since it will be aware of its neighbors' colors before choosing one itself).
\end{proof}

Recall that we use the notation $\NUMFRAGSDET = 240000$. We show that the number of phases needed to reduce the number of fragments to one is at most $\lceil \log_{\NUMFRAGSDET/(\NUMFRAGSDET -1)} n \rceil + \NUMFRAGSDET$. 
Correctness immediately follows as we implement GHS. 
 We first argue that in each phase of the algorithm where there is initially a sufficient number of fragments $\NUMFRAGSDET$, the number of fragments is reduced by a constant factor of $\NUMFRAGSDET/(\NUMFRAGSDET - 1)$.  We then show that in an additional $\NUMFRAGSDET$ phases, we can reduce the number of fragments to one. As we only add MST edges and all nodes will be in this one fragment, the final fragment represents the MST of original graph.

Let $\DETIMPPHASE$ represent the phase by which the number of fragments at the beginning of the phase is less than $\NUMFRAGSDET$. We eventually show that $\DETIMPPHASE = \lceil \log_{\NUMFRAGSDET/(\NUMFRAGSDET - 1)} n \rceil$. We first argue that in every phase up to $\DETIMPPHASE$, the number of fragments is reduced by a constant factor. We do this by considering an arbitrary phase $i$ and identifying a set of fragments in that phase that are guaranteed to merge into other fragments, thus ``being lost'' or ``not surviving'' in that phase. We show that this set is at least a constant fraction of the total set of fragments that existed at the start of the phase. 

Consider an arbitrary phase $i$ such that at the beginning of the phase there exists a set $\mathcal{F}_i$ of fragments and define $F_i = |\mathcal{F}_i|$. Furthermore, define the supergraph $H$ as the undirected graph where the nodes are the set $\mathcal{F}_i$ and the edges are the valid MOEs between the different fragments, i.e., the graph obtained after pruning MOEs in step (ii) of the phase. In the subsequent analysis we use nodes and fragments interchangeably in the context of graph $H$. We now show that the number of Blue fragments (which by the algorithm are all merged into other fragments) constitute a sufficiently large constant fraction of $\mathcal{F}_i$.

\begin{lemma}\label{lem:fraction-blue-frag-merge-in-phase}
Let $H'$ be a connected subgraph of $H$. If $|H'| \geq 342$, then at least $\lfloor |H'|/342 \rfloor$ of the fragments are Blue.
\end{lemma}

\begin{proof}
Let $B$, $R$, $O$, $Bl$, and $G$ represent the number of Blue, Red, Orange, Black, and Green fragments in $H'$, respectively. We can get a lower bound on the number Blue fragments in $H'$ by assuming that the adversary caused the maximum number of fragments to be colored other colors.

For a fragment to be colored Red, at least one of its neighbors must be colored Blue. Furthermore, any given Blue fragment can be a neighbor to at most $4$ Red fragments. Thus, the number of Red fragments is upper bounded by $4 B$. Similar arguments show that the number of Orange fragments, Black fragments, and Green fragments are upper bounded by $4^2 B$, $4^3 B$, and $4^4 B$, respectively.  

Thus, the number of fragments in $H'$, $|H'| = B + R + O + Bl + G \leq B + 4B + 4^2B + 4^3B + 4^4B \leq B(4^5 - 1)/3 \leq 342 B$. Thus, we see that $B \geq \lfloor |H'|/342 \rfloor$.
\end{proof}

The above lemma by itself is insufficient to show that the required number of fragments are removed in each phase. The reason is that $H$ may consist of a set of disjoint connected subgraphs. 
Let us assume that $|H| \geq \NUMFRAGSDET$. Let $S$ denote the set of all disjoint connected subgraphs (i.e., connected components) in $H$. Now, either $|S| \geq |H|/ 342^2$ or $|S| < |H|/ 342^2$. We show that in either case, the number of fragments that survive the current phase is $(\NUMFRAGSDET - 1) |H| / \NUMFRAGSDET$.

If $|S| \geq |H|/342^2$, then since each subgraph in $S$ contains at least one Blue fragment which disappears in the phase, the total number of fragments that survive the phase is at most $|H| - |H|/ 342^2 \leq (\NUMFRAGSDET - 1) |H| / \NUMFRAGSDET$.\footnote{It is easy to see that each subgraph contains a Blue fragment because Blue is the highest priority color and so the first fragment that colors itself in any given subgraph colors itself Blue.}

Now let us look at the situation where $|S| < |H|/342^2$. Divide $S$ into the sets $S_1$ and $S_2$ which contain the disjoint connected subgraphs of $H$ which have $<342$ fragments and $\geq 342$ fragments, respectively. Observe that $S = S_1 \bigcup S_2$. (It is easy to see that $|S_2| \geq 1$ since otherwise if all subgraphs belonged to $S_1$, there would be less than $|H|$ total fragments.) We now show that a sufficient number of fragments in the subgraphs in set $S_2$ are Blue fragments, thus resulting in a sufficient number of fragments being removed in the phase. Let us lower bound how many fragments are present in the subgraphs in set $S_2$. Recall that $|S_1| \leq |S|$, we are considering the situation where $|S| < |H|/342^2$, and each subgraph in $|S_1|$ can have $< 342$ fragments. We lower bound the fragments in $S_2$ by pessimistically ignoring the less than $342$ fragments from each of the $|S_1|$ subgraphs (recall that $|S_1| \leq |S|$), i.e., the number of fragments in $S_2$ is 

\begin{align*}
    &\geq |H| - (|H|/342^2)\cdot 342 \text{ (since the total number of subgraphs is at most } |H|/342^2 )\\
    &= 341 |H|/ 342. 
\end{align*}

We now lower bound the number of Blue fragments in subgraphs in $S_2$. Let the subgraphs in $S_2$ be denoted by $H_1, H_2, \ldots, H_{|S_2|}$. Since each subgraph in $S_2$ is of size at least $342$, we can use Lemma 4. Now, the number of Blue fragments in $S_2$ is 
\begin{align*}
&= \sum\limits_{i=1}^{|S_2|} \lfloor |H_i|/342 \rfloor \\ 
&\geq \sum\limits_{i=1}^{|S_2|} (|H_i|/342 - 1) \\
&= 1/342 \cdot (\sum\limits_{i=1}^{|S_2|} |H_i|) - |S_2| \\
&\geq 1/342 \cdot (341|H_i|/342) - |S_2| \\
&\geq 1/342 \cdot (341|H_i|/342) - |H|/342^2 \text{ (since } |S_2| \leq |S| < |H|/342^2)\\
&=340 |H|/ 342^2.
\end{align*}

Recall that all Blue fragments do not survive a phase. Thus, the number of fragments that survive the current phase is 
\begin{align*}
\leq |H| - 340 |H|/342^2 \\
\leq (\NUMFRAGSDET - 1) |H| / \NUMFRAGSDET.
\end{align*}

Thus, in both situations, we see that the number of fragments that survive the present phase is upper bounded as desired.

\begin{lemma}\label{lem:constant-fraction-reduction}
After $\lceil \log_{\NUMFRAGSDET/(\NUMFRAGSDET - 1)} n \rceil$ phases, there are at most $\NUMFRAGSDET$ fragments at the beginning of the phase.
\end{lemma}

An easy observation is that whenever $F_i \geq 2$, at least one fragment is Blue and will not survive the phase. Thus, if there are at most $\NUMFRAGSDET$ fragments at the beginning of the phase, then running an additional $\NUMFRAGSDET$ phases guarantees that only one fragment will remain. Initially, each node is a fragment by itself and over the course of the algorithm, only possible MST edges are added to any fragment. Thus, we have the following lemma.

\begin{lemma}\label{lem:det-correctness guarantee}
Algorithm~$\textsc{Deterministic-MST}$ correctly outputs an MST after $\lceil \log_{\NUMFRAGSDET/(\NUMFRAGSDET - 1)} n \rceil + \NUMFRAGSDET$ phases.
\end{lemma}

We analyze the running time and awake time for each node in each phase. Since there are $O(\log n)$ phases, it is easy to then get the overall running and awake times. 

Let us look at each step of a phase individually. In step (i), all nodes participate in two instances of $\DOWNCAST(n)$, one instance of $\UPCASTMIN(n)$, and one instance of $\TRANSMITADJACENT(n)$. Each of these procedures takes $O(1)$ awake time and $O(n)$ run time. 

In step (ii), all nodes participate in one instance of $\UPCASTMIN(n)$, one instance of $\TRANSMITADJACENT(n)$, and two instances of $\TS(\cdot,\cdot,n)$. Again, each of these procedures takes $O(1)$ awake time and $O(n)$ run time.

\begin{sloppypar}
In step (iii), all nodes participate in two instances of $\DOWNCAST(n)$, two instances of $\NEIGHBORAWARE(n)$, two instances of $\TRANSMITADJACENT(n)$, and one instance of $\AWAKECOLORING(n,N)$, and two instances of $\REORIENTFRAG(n)$. $\DOWNCAST(n)$ and $\TRANSMITADJACENT(n)$ each take $O(1)$ awake time and $O(n)$ running time. $\NEIGHBORAWARE(n)$ is a variant of $\UPCASTMIN(n)$ and from its description, it is easy to see that it too takes $O(1)$ awake time and $O(n)$ run time. $\REORIENTFRAG(n)$ was described and analyzed in the previous section and took $O(1)$ awake time and $O(n)$ running time. We see that $\AWAKECOLORING(n,N)$ consists of $O(\log^4 n)$ stages, where each stage takes $O(n)$ run time. However, each node is awake for and participates in at most $5$ such stages from Lemma~\ref{lem:awake-coloring-works}. In each stage, nodes that are awake participate in one instance of $\TRANSMITADJACENT(n)$, one instance of $\NEIGHBORAWARE(n)$, and one instance of $\DOWNCAST(n)$. Thus, each stage contributes $O(1)$ awake time for participating nodes and $O(n)$ running time. Thus, $\AWAKECOLORING(n,N)$ has $O(1)$ awake time and $O(n\log^4 n)$ running time. Thus, we have the following lemma.
\end{sloppypar}

\begin{lemma}\label{lem:det-mst-complexity-guarantees}
Each phase of Algorithm~$\DETMST$ takes $O(1)$ awake time and $O(n\log^4 n)$ running time.
\end{lemma}

Thus, combining Lemma~\ref{lem:det-correctness guarantee} and Lemma~\ref{lem:det-mst-complexity-guarantees}, we get the following theorem.

\begin{theorem}
Algorithm~$\DETMST$ is a deterministic algorithm to find the MST of a given graph in $O(\log n)$ awake time and $O(n \log^5 n)$ running time.
\end{theorem}

\textbf{Remark.} The coloring procedure, $\AWAKECOLORING(n,N)$, is the main reason for the large run time. As we noted near the beginning of this section, we can replace this procedure with one that can accomplish this deterministically
in $O(\log^*n)$ run time even in the traditional model (see e.g., 
\cite{PanduranganRS17}). However, we suffer an overhead of an $O(\log^*n)$ multiplicative factor in the awake time. As a result, using this modified procedure would allow us to get the following theorem.

\begin{theorem}
There exists a deterministic algorithm to find the MST of a given graph in $O(\log n \log^* n)$ awake time and $O(n \log n \log^* n)$ run time.
\end{theorem}

%% file: mst-trade-offs.tex

In this section, we present an algorithm to create an MST that shows a trade-off between its running time and awake time. 
We use a few procedures listed in Section~\ref{sec:prelims} and additional procedures mentioned below. First we give a brief overview of the algorithm, then describe the additional procedures, before finally giving a detailed explanation of the algorithm along with analysis.

\textbf{Brief Overview.}
We now describe algorithm $\MSTTRADEOFF$, which finds the MST of the given graph with high probability in $\Tilde{O}(D + 2^k + n/2^k)$ running time and $\Tilde{O}(n/2^k)$ awake time, where integer $k$ is an input parameter to the algorithm. When $k \in [\max \lbrace \lceil 0.5\log n \rceil, \lceil \log D \rceil \rbrace, \lceil \log n \rceil]$, we obtain useful trade-offs. In essence, it is an awake time efficient version of the algorithm from Chapter 7 of~\cite{dnabook} (itself a version of the optimized version of the algorithm from~\cite{GKP98}) adapted to reduce the awake complexity. We describe it as a three stage algorithm. 

In stage one, we elect a leader node among all the nodes. We then construct a Breadth First Search (BFS) tree rooted at this leader, which will be used later on. In stage two, we switch gears and have all nodes perform the controlled version of the GHS algorithm for $k-1$ phases until $\leq n/2^k$ fragments are formed, each of size $\geq 2^k$ and diameter $\leq 5 \cdot 2^k$. In stage three, each node  now uses the BFS tree formed in stage one to send its MOE (inter-fragment MOE) for each of the at most $n/2^k$ fragments with corresponding node IDs and fragments IDs to the leader using pipelining. Using the red rule to prevent cycles, we ensure this does not take too long. The leader then performs a local computation to determine which $O(n/2^k)$ edges between the $O(n/2^k)$ fragments are a part of the MST and send messages about those edges down to the respective nodes.

\textbf{Preliminaries.} We describe how to elect a leader and construct a BFS by relying on procedures from~\cite{DH22}. 

We make use of procedures $\BUILDMSC$ and $\SAF$, parameterized appropriately, in conjunction with procedures described below to elect a leader and construct the BFS in a manner that reduces the awake time. In order to run the procedures $\BUILDMSC$ and $\SAF$, all nodes need to know a constant factor upper bound on $n$.\footnote{In~\cite{DH22}, in addition to knowledge of the upper bound on $n$, they also assume that upper bounds on the diameter $D$ and maximum degree $\Delta$ are known. However, the algorithm can be modified to take the same running time and energy complexity while only requiring the knowledge of $n$~\cite{V23}.} Nodes also need unique IDs, an assumption already made in the current paper. Now, the $\BUILDMSC$ procedure builds a multi-level clustering over the graph and $\SAF$ uses these clusters to simulate a simple algorithm $\mathcal{A}$ that takes $T_{\mathcal{A}}$ running time in $\Tilde{O}(T_{\mathcal{A}})$ running time and $\Tilde{O}(\varepsilon + T_{\mathcal{A}}/R^{\ell})$ awake time, where $\varepsilon$ represents the energy cost of the ``greedy psychic'' version of the algorithm (essentially if each node only had to stay awake for the bare minimum number of rounds to transmit messages for the algorithm) and $R$ and $\ell$ are input parameters to the algorithm. 

The algorithms from~\cite{DH22} hold for the radio network model in the ``OR'' version of the model, defined as follows. In every round that a node is awake and listens, it receives an arbitrary message from one of its neighbors, assuming that some non-empty subset of its neighbors chose to send messages in that round. If all of its neighbors are asleep in that round, then it does not receive a message. It is easy to see that their algorithm can be run directly in the awake model, where in every round a node is awake, it receives all messages sent by its neighbors in that round. Furthermore, the upper bounds on running time and energy complexity in the OR model serve as upper bounds on the running time and awake complexity in the awake model, respectively.

We may consider two simple algorithms, $\SIMPLELE$ and $\SIMPLEBFS$, that work as follows. $\SIMPLELE$ has each node choose an rank uniformly at random in $[1,n^4]$ and transmits it to its neighbors. Subsequently, every time a node hears of a rank higher than the current highest rank it has seen, it updates this value, and transmits this new rank to its neighbors. Finally, after $O(D)$ rounds, the node with the highest rank broadcasts its ID to all nodes. The procedure has running time $O(D)$ when $D$ is known.\footnote{The $\BUILDMSC$ algorithm of~\cite{DH22} can be modified to learn a constant factor upper bound on $D$~\cite{V23}.} It can be proven that each node, when running this algorithm, needs to transmits a new rank $O(\log n)$ times with high probability, i.e., each node only sends and receives messages in $O(\log n)$ rounds. Thus, the $\varepsilon$ for this algorithm is $\Tilde{O}(1)$. $\SIMPLEBFS$ is parameterized by some root and has this root set its distance from the root to $0$ and transmits a message with its value to its neighbors. Each node, upon receiving such a message with distance $i$ from some non-empty subset of its neighbors, chooses one of those nodes as its parent and informs it via a message, sets its distance from root to $i+1$ and transmits this value to the neighbors it did not receive a message from. It is easy to see that the running time of the algorithm is $O(D)$ rounds when $D$ is known and each node needs to send and receive messages in $O(1)$ rounds. Thus, the $\varepsilon$ for this algorithm is $O(1)$. By choosing appropriate $R$ and $\ell$ values when calling these algorithms from $\SAF$, we can elect a leader and construct a BFS from it in $\Tilde{O}(1)$ awake time and $\Tilde{O}(D)$ running time (see Theorems~14,~15, and~16 in~\cite{DH22}).

Let us call the entire process of constructing the multi-level clustering and using it to elect a leader and construct a BFS tree from that leader algorithm~$\LEBFS$.

Thus, we get the following lemma.

\begin{lemma}\label{lem:le-bfs-fast-awake}
Algorithm~$\LEBFS$ elects a leader and grows a BFS tree from it with high probability in $\Tilde{O}(D)$ running time and $\Tilde{O}(1)$ awake time.
\end{lemma}


\textbf{Detailed Algorithm.}
In stage one, we elect a leader and construct a BFS tree from this leader, which will be used later. We have nodes run $\LEBFS$. 

We now describe stage two, which corresponds to the controlled GHS stage of the algorithm. We explain how to modify $\RANDMST$, described in Section~\ref{sec:mst-optimal-awake-time}, into the required controlled GHS version. 
Throughout the algorithm in Section~\ref{sec:mst-optimal-awake-time}, we used procedures parameterized by $n$. Here, we use algorithms parameterized by $5 \cdot 2^i$, where $i$ corresponds to the current phase, resulting in procedures with reduced running times. 

Now, we describe additional changes needed within each phase of $\RANDMST$ in order to adapt it for stage two. Recall that in controlled GHS, there are two things to note. First, the fragments that actually participate in a given phase $i$ are limited to those with diameter $\leq 2^i$. Second, once these fragments find their MOEs, a maximal matching is found on the supergraph consisting of fragments as nodes and MOEs as edges between nodes. Only those fragments that are matched together may merge.

To implement the first change, modify step (i) of each phase of $\RANDMST$ as follows. An invariant that holds across phases is that each fragment at the beginning of phase $i$ does not have a diameter more than $5 \cdot 2^i$. Each node also knows its distance from the root. Have each node in the graph run a variant of $\UPCASTMIN(5 \cdot 2^i)$ where the \emph{maximum} distance from the root is propagated to the root. The root can then make a decision on whether to participate in the given phase. All nodes subsequently participate in $\DOWNCAST(5 \cdot 2^i)$ where the root of the fragment informs all nodes in that fragment whether they will participate in the current phase or not. 

To implement the second change, modify steps (ii) and (iii) of each phase of $\RANDMST$ as follows. Step (ii) of each phase of $\RANDMST$ is entirely replaced by the following. Consider the disjoint set of subgraphs formed by the union of fragments and their MOEs. Consider one such subgraph. Now, in order to perform a matching on the fragments, we first use the well known combination of Linial's graph coloring algorithm~\cite{L92} and the Kuhn-Wattenhofer color reduction technique~\cite{KW06}, 
which when applied to a rooted tree of $n$ nodes, gives a $3$ coloring on the tree in $O(\log^* n)$ rounds. This tree is the supergraph formed by taking each fragment in the subgraph as a node and considering the MOEs (taken in the opposite direction) as edges between the nodes. In the subgraph, there will always be two adjacent fragments with chosen MOEs leading to each other. We may choose the minimum ID fragment among these two as the root of the tree. In order to simulate one round of coloring on the supergraph, a combination of $\DOWNCAST(5 \cdot 2^i)$, $\TRANSMITADJACENT(5 \cdot 2^i)$, and $\UPCASTMIN(5 \cdot 2^i)$ can be used. From this $3$ coloring on the supergraph, we can obtain a maximal matching in $O(1)$ rounds as follows. Have each supernode with color $1$ inform its parent that they are now matched. Then have each unmatched supernode with color $2$ query its parent and ask to match with it. If the parent is unmatched, it replies as such and both supernodes are now matched. In case multiple children request the same parent for a match, the parent chooses the fragment with the minimum ID (this is easy to do as part of $\UPCASTMIN(5 \cdot 2^i)$ run by the nodes of the parent's fragment). Repeat the previous step with unmatched supernodes with color $3$. As before, each round on the supergraph can be implemented through a combination of a constant number of instances of $\DOWNCAST(5 \cdot 2^i)$, $\TRANSMITADJACENT(5 \cdot 2^i)$, and $\UPCASTMIN(5 \cdot 2^i)$. In order to merge the fragments, we may run step (iii) of each phase of $\RANDMST$ as is by considering matched children and parent fragments as corresponding Tails and Heads fragments, respectively. 

Finally, we come to stage three of the algorithm. We have a disjoint set of $\leq n/2^k$ fragments, each of diameter at most $5 \cdot 2^k$, for which we must find the remaining at most $n/2^k$ MOEs. We make use of the BFS tree formed over the original graph in stage one and run the pipelining algorithm from~\cite{dnabook}, slightly modified to optimize awake time. We describe the modified procedure here. 
Each node $v$ maintains two lists, $Q$ and $U$. Initially, $Q$ contains only the inter-fragment edges adjacent to $v$ and corresponding fragment IDs and $U$ is empty. In each round that $v$ is awake, $v$ sends the minimum-weight edge in $Q$ that does not create a cycle with the edges in $U$ to its parent along with corresponding fragment IDs and moves this edge from $Q$ to $U$. 
If $Q$ is empty, $v$ does not send a message. 
The parent after receiving an edge from a child, adds the edge to its $Q$ list. 
To ensure that each node is only awake for $O(n/2^k)$ rounds, we have nodes stay awake according to a schedule. All nodes are aware of the first round of stage three, call it round $1$. Each node at depth $d$ stays awake from round $D - d$ to round $D-d + n/2^k$: in round $D-d$, the node listens for messages from its children, and in subsequent rounds it transmits messages according to the procedure described. Notice that since we start from round $1$, nodes at depth $D$, if any, do not have a round where they listen to their children in the BFS tree, since none exist. 
Once all messages reach the root of the BFS tree by round $D + n/2^k$, the root can then calculate the MOEs for the fragments. Nodes then participate in a reverse of the schedule described above to send down the at most $n/2^k$ MOEs with corresponding fragment IDs to all the nodes, thus ensuring all nodes know their edges in the MST.

\textbf{Analysis.}
We now prove that the algorithm correctly outputs the MST of the original graph and the running time and awake time are as desired. We first argue about the correctness below.
\begin{lemma}\label{lem:msttradeoff-correctness}
Algorithm~$\MSTTRADEOFF$ correctly outputs an MST with high probability.
\end{lemma}

\begin{proof}[Proof Sketch]
In stage one, from Lemma~\ref{lem:le-bfs-fast-awake}, we see that a leader is elected and a BFS tree is grown from it with high probability.

In stage two, since we know the bound on the maximum diameter of any fragment in each phase, the modifications to the GHS with appropriately parameterized instances of $\UPCASTMIN$, $\DOWNCAST$, $\TRANSMITNEIGHBOR$, and $\TRANSMITADJACENT$ guarantees that at most $n/2^k$ fragments, each of diameter at most $5 \cdot 2^k$, are formed at the end of the stage (this is because it correctly simulates the controlled GHS algorithm as described in~\cite{dnabook}).

In stage three, we can see that by collecting MOEs, starting from the nodes furthest away from the root in the BFS, and applying the red rule the way we do, we ensure that at least the set of edges belonging to the MST over the supergraph consisting of fragments and their MOEs are sent to the root of the BFS. The root can then calculate the exact edges, which are at most $O(n/2^k)$ edges, and send them back down to all nodes in the tree.
\end{proof}

We now argue about the awake times and running times of nodes in each stage.

From Lemma~\ref{lem:le-bfs-fast-awake}, we see that stage one of $\MSTTRADEOFF$ takes $\Tilde{O}(D)$ running time and $\Tilde{O}(1)$ awake time.

\begin{lemma}\label{lem:msttradeoff-stage-two-time}
Stage two of $\MSTTRADEOFF$ takes $\Tilde{O}(2^k)$ running time and $\Tilde{O}(1)$ awake time.
\end{lemma}

\begin{proof}
Stage two consists of $k = O(\log n)$ phases, where each phase $i$ consists of $O(\log^* n)$ instances of $\UPCASTMIN$, $\DOWNCAST$, $\TRANSMITNEIGHBOR$, and $\TRANSMITADJACENT$ parameterized by $5 \cdot 2^i$, which is upper bounded by $O(2^k)$. From, Observations~\ref{obs:downcast},~\ref{obs:upcastmin},~\ref{obs:trasnmit-neighbor}, and~\ref{obs:transmit-adjacent}, we get the desired bounds.
\end{proof}

\begin{lemma}\label{lem:msttradeoff-stage-three-time}
Stage three of $\MSTTRADEOFF$ takes $O(D + n/2^k)$ running time and $O(n/2^k)$ awake time.
\end{lemma}

\begin{proof}
Pipelining takes $O(D + n/2^k)$ running time and $O(n/2^k)$ awake time. Once the root of the BFS tree calculates the MOEs, it takes an additional $O(D + n/2^k)$ running time and $O(n/2^k)$ awake time to transmit these values down to all nodes.
\end{proof}

Putting together Lemmas~\ref{lem:msttradeoff-correctness},~\ref{lem:le-bfs-fast-awake},~\ref{lem:msttradeoff-stage-two-time}, and~\ref{lem:msttradeoff-stage-three-time}, we have the following theorem.

\begin{theorem}\label{the:msttradeoff-main-theorem}
Algorithm~$\MSTTRADEOFF$ is a randomized algorithm to find the MST of a graph with high probability and takes $\Tilde{O}(D + 2^k + n/2^k)$ running time and $\Tilde{O}(n/2^k)$ awake time, where $k$ is an input to the algorithm.
\end{theorem}

%% file: conclusion.tex
We presented distributed algorithms for the fundamental MST problem that are optimal with respect
to awake complexity. We also showed that   there is an inherent trade-off bottleneck between awake and round complexities of MST. In other words, one cannot attain
optimal complexities simultaneously under both measures. 
We also presented an algorithm that shows a trade-off between awake complexity and round complexity, complementing our
trade-off lower bound. 
Interesting lines of future work including designing awake-efficient algorithms for other fundamental global problems such as shortest path and minimum cut.